\documentclass[9pt]{article}

\usepackage[utf8]{inputenc}
\usepackage[american]{babel}
\usepackage{xcolor}
\usepackage{lineno}
\usepackage{xspace}
\usepackage{etoolbox}
\usepackage[autolanguage]{numprint}
\usepackage{url}
\usepackage{algorithm}
\usepackage{algpseudocode}
\usepackage{subcaption}
\usepackage{booktabs}
\usepackage{amsmath}
\usepackage{amsthm}
\usepackage{authblk}
\usepackage{amssymb}
\usepackage{titling}
\usepackage{multicol}
\usepackage{hyperref}
\usepackage{tabularx}
\usepackage[colorinlistoftodos,prependcaption,textsize=footnotesize]{todonotes}

\newtheorem{theorem}{Theorem}[section]
\newtheorem{lemma}[theorem]{Lemma}

\theoremstyle{remark}

\theoremstyle{definition}

\newcommand{\budget}{\mathcal{G}\xspace}
\newcommand{\oopt}{\ensuremath{\textit{Opt}}\xspace}
\newcommand{\ucas}{\ensuremath{\texttt{UCAS}}\xspace}
\newcommand{\calI}{\mathcal{I}}
\newcommand{\idleP}[1]{\mathcal{P}_{\text{idle}}^{#1}}
\newcommand{\workP}[1]{\mathcal{P}_{\text{work}}^{#1}}
\newcommand{\PP}[1]{\mathcal{P}^{#1}}
\newcommand{\etal}{et al.\xspace}
\newcommand{\algvar}[1]{\textsf{#1}}
\newcommand{\cluster}[1]{\texttt{#1}}

\newcommand{\EST}{\texttt{ASAP}\xspace}

\usepackage{tikz-timing}
\usepackage{tikz}
\usetikzlibrary{calc}
\usetikzlibrary{patterns}
\usetikzlibrary{arrows,shapes}
\usetikzlibrary{shapes.geometric}
\usetikzlibrary{decorations.pathreplacing}
\usetikztiminglibrary[new={char=Q,reset char=R}]{counters}
\pgfdeclarelayer{background}
\pgfdeclarelayer{foreground}
\pgfsetlayers{background,main,foreground}

\newcommand{\timeline}[3]{
\draw[thin, color=black,->] (#1,#3) -- (#2,#3) node[below=-0.5pt, ] {\scriptsize{time}};
}
\newcommand{\interval}[3]{
\draw[thin, color=red] ($(#1,#3)+(0,-0.6)$) -- ($(#1,#3)+(0,1.2)$);
\draw[thin, color=red] ($(#2,#3)+(0,-0.6)$) -- ($(#2,#3)+(0,1.2)$);
}

\usepackage[normalem]{ulem}

\newtoggle{TR}
\toggletrue{TR}

\newtoggle{LV}
\toggletrue{LV} 

\newtoggle{showOld}
\togglefalse{showOld}

\newtoggle{showChanges}
\togglefalse{showChanges}

\newcommand{\old}[1]{\iftoggle{showOld}{\textcolor{red}{\sout{#1}}}{}}
\newcommand{\new}[1]{\iftoggle{showChanges}{\textcolor{blue}{#1}}{\textcolor{black}{#1}}}

\pretitle{%
  \begin{center}\rule{\textwidth}{1pt}\par\vspace{0.5em}%
  \Large\bfseries}
\posttitle{\par\rule{\textwidth}{1pt}\end{center}}
\title{Carbon-Aware Workflow Scheduling \\ with Fixed Mapping and Deadline Constraint}
\date{}
\author{
  Dominik Schweisgut\thanks{Humboldt-Universität zu Berlin, Germany, \texttt{dominik.schweisgut@kit.edu}, now at Karlsruhe Institue of Technology (KIT)}
  \and
  Anne Benoit\thanks{ENS Lyon and IUF, France \& IDEaS, Atlanta, USA \texttt{Anne.Benoit@ens-lyon.fr}}
  \and
  Yves Robert\thanks{ENS Lyon, France, \texttt{Yves.Robert@ens-lyon.fr}}
  \and
  Henning Meyerhenke\thanks{Karlsruhe Institute of Technology (KIT), Germany, \texttt{meyerhenke@kit.edu}}
}

\begin{document}

%
%
%

   

\maketitle

\vspace{-3em}
\begin{abstract}


Large data and computing centers consume a significant share of the world's energy consumption.
A prominent subset of the workloads in such centers are workflows with interdependent tasks,
usually represented as directed acyclic graphs (DAGs).
To reduce the carbon emissions resulting from executing such workflows in centers with a mixed 
(renewable and non-renewable) energy supply, it is advisable to move task executions to
time intervals with sufficient green energy when possible. 
To this end, we formalize the above problem as a scheduling problem with a given mapping and
ordering of the tasks. We show that this problem can be solved in polynomial time in the uniprocessor
case. For at least two processors, however, the problem becomes NP-hard.
Hence, we propose a heuristic framework called CaWoSched that combines several greedy
approaches with local search. To assess the 16 heuristics resulting from different combinations,
we also devise a simple baseline algorithm and an exact ILP-based solution.
Our experimental results show that our heuristics provide significant savings in carbon
emissions compared to the baseline.

\end{abstract}


\section{Introduction}
%
The number of large and geo-distributed data and computing centers grows rapidly and so is the amount of data processed by them. 
Their services have become indispensable in industry and academia alike. 
Yet, these services result in a globally significant energy consumption as well as carbon emissions due to computations and data transfer, see for example~\cite{AhmedBA21review} for concrete numbers.
Since the combined carbon footprint of all data/computing centers on the globe is even higher than that of air traffic~\cite{lavi23measuring},
reducing this footprint is of immense importance, both from an ecological, political, and
economical perspective.
On the technical side, data/computing centers have started to use a mix of different power sources, 
giving priority to lower carbon-emitting technologies (solar, wind, nuclear) over higher ones (coal, natural gas).
This raises new challenges and opportunities for HPC (High Performance Computing) scientists; \
for example, designing efficient scheduling algorithms was already a complicated task when computer platforms had only a single power source -- and thus the same level of carbon emissions at each point in time.
This task becomes even more difficult when a mix of power sources leads to different carbon emissions over time:
in addition to optimizing only standard performance-related objectives, one important new objective is 
to optimize the total amount of carbon emissions induced by the execution of all applications in a particular data/computing center.

%
%
Many workloads in a data/computing center, not only but in particular in a scientific context, can be seen as workflows 
consisting of individual tasks with input/output relations. The algorithmic task we focus on in this paper is
to schedule such workflows (abstracted as directed acyclic task graphs) on a parallel platform within some deadline. 
To focus on carbon footprint minimization, we assume that the mapping and ordering of all tasks and communications are already given, for instance as the result of executing the  de-facto standard HEFT algorithm~\cite{topcuoglu2002performance}. 
A somewhat similar approach of assuming a fixed mapping can be found in the literature, e.\,g., for energy- and reliability-aware 
task scheduling and frequency scaling in embedded systems~\cite{Pop07scheduling}.
%

In our simplified setting with a given mapping, minimizing the total execution time, or makespan, can easily be achieved in linear time with the ASAP greedy algorithm:
simply execute each task as soon as all preceding tasks and corresponding communications have completed. 
The computing platform that we target instead has a time-varying amount of green energy available, for example due to solar and/or wind power
produced for its data/computing center.
Moreover, the carbon emissions can vary from processor to processor due to the latter's different power demands,
making the platform completely \emph{carbon-heterogeneous}.

The problem becomes combinatorial: while we cannot change the mapping nor the ordering of the tasks on each processor, 
we can shift the tasks (and the corresponding communication operations) to benefit from lower-carbon intervals, while enforcing the deadline.
Previous studies have shown that exploiting lower-carbon intervals can be very beneficial, see e.\,g.,~\cite{wiesnerLetWaitAwhile2021e}.
Yet, this previous line of research has either not focused on scheduling individual workflows
or focused on reducing energy consumption~\cite{versluis2022taskflow,bader2025predicting} rather than carbon emission.
As we outline in more detail in Section~\ref{sec.related},
carbon-aware scheduling algorithms are still in their infancy.

\vspace{-0.2cm}
\paragraph*{Contributions.}
The main contributions of this paper are both theoretical and practical. On the theory side, we lay the foundations 
for the problem complexity, with \new{(i)} a sophisticated fully-polynomial time dynamic programming algorithm 
for the single processor case, \old{a proof of strong NP-completeness for a simple 2-processor instance 
with independent tasks and carbon-homogeneous processors} \new{(ii) a proof of strong NP-completeness for a simplified instance
with 2 or more processors, independent tasks, and carbon-homogeneous processors}, and \new{(iii)} the formulation of the general problem 
\old{into} \new{as} an integer linear program (ILP). On the practical side, we 
design efficient algorithms that greatly decrease the total carbon cost compared to a standard 
carbon-unaware competitor; for small instances, \old{we show that our algorithms achieve a close performance to the ILP.}
\new{our experimental results indicate that our algorithms achieve a quality that is close to the optimal one derived from the ILP.}

\vspace{-0.2cm}
\paragraph*{Outline.} The rest of the paper is organized as follows. Section~\ref{sec.related} surveys related work. In Section~\ref{sec.framework},
we detail the framework.  Section~\ref{sec.complexity} is devoted to complexity results.
We introduce new carbon-aware algorithms in Section~\ref{sec.algos} and assess their performance through an extensive set of simulations in Section~\ref{sec.expes}.
Finally, we give concluding remarks and hints for future work in Section~\ref{sec.conclusion}.
\old{\textbf{Note that material omitted due to space constraints can be found in the appendix of the full version of this paper~\cite{zenodo}.}}

\section{Related Work}
\label{sec.related}


Carbon-aware computing has received increasing attention in the past few years, acknowledging the clearly non-negligible share data/computing centers have on mankind's carbon footprint --
as well as the need for action to reduce the emissions given the rapid increase in data to be processed~\cite{cao2022toward}.
Most works in this more general line of research retain a high-level workload perspective and thus do not consider the concrete task of workflow scheduling.
For example, a carbon-aware load balancing algorithm to reduce the carbon footprint of geo-distributed 
data centers considers abstract workloads, not interdependent tasks~\cite{MahmudI16distributed}. 
It uses the \emph{alternating direction method of multipliers} to move workloads to locations with lower carbon intensity.
On a similar granularity, global cloud providers use scheduler-agnostic workload shifting to 
less carbon-intensive data centers, depending on the projected availability of green energy in suitable locations and time intervals~\cite{radovanovic2022carbon}.
A high-level workload perspective and a similar objective is used by Hall et al.~\cite{hall2024carbon}.
They devise a two-phase approach of (i) day-ahead planning based on historical data and (ii) real-time job placement
and scheduling. As they consider abstract workloads and not workflows with interdependencies, their approach
is not directly comparable, either.
Finally, while Breukelmann \etal~\cite{breukelman2024carbon} model interconnected 
data centers as a weighted graph, they still consider unrelated batch compute jobs as the workload.
They formalize the optimal allocation problem in this setting as a single-leader multiple-follower
Stackelberg game and suggest an ad-hoc algorithm (which is not applicable in our setting) to solve it.


Regarding workflow scheduling in general, we refer the interested reader to 
surveys~\cite{adhikari2019, liu2018survey} and a monograph~\cite{sinnen2007task} for a broader overview.
One possible way to categorize workflow scheduling algorithms is to distinguish
online algorithms (which do not know the complete workflow when taking decisions for
tasks) and plan-based algorithms. 
Since this paper proposes a plan-based algorithm, we focus on the latter.
Even rather simplistic versions of plan-based scheduling are NP-hard~\cite{GareyJohnson},
which motivates the use of heuristics for real-world applications. 
Two common approaches are list- and partitioning-based heuristics.
HEFT (heterogeneous earliest finish time)~\cite{topcuoglu2002performance} is a very influential and still popular 
list-scheduling algorithm that has seen numerous extensions and variations over the 
years~\cite{arabnejad2013list,SHI2006665,SANDOKJI2019482,pheft,samadi2018eheft}.
It has two main phases that (i) assign priorities to tasks and (ii) then assign tasks
to processors based on the priorities from the previous phase.

Partitioning-based scheduling heuristics, in turn, group tasks into blocks and assign
these blocks to processors, see e.\,g.~\cite{Ozkaya19-IPDPS,kulagina2024mapping,viil2018framework}.
This aggregation step helps in reducing the complexity of dealing with individual task assignments
in large-scale workflows.

Two prominent algorithms for energy-efficient workflow scheduling are 
GreenHEFT~\cite{durilloMultiobjectiveEnergyefficientWorkflow2014b} and 
MOHEFT~\cite{durilloParetoTradeoffScheduling2015}. Both heuristics
optimize \emph{where} tasks are scheduled in order to save energy.
Similar to one type of our heuristics, TaskFlow~\cite{versluis2022taskflow} 
exploits \emph{slack} in workflows, i.\,e., it takes advantage of tolerable delays
by executing the corresponding tasks on more energy-efficient hardware. 
Yet, they all do not optimize for carbon emissions and thus do not consider \emph{when}
tasks should run in order to exploit green energy availability.

The importance of reducing carbon emissions has led to a number of papers working on this goal.
Wen et al.~\cite{wenRunningIndustrialWorkflow2021b}, for example, propose a genetic algorithm for
adaptive workflow mapping whose main rationale is to move tasks between geographically distributed
data centers -- depending on their energy mix. The approach only provides a mapping of tasks to 
data centers, but no task starting times. Moreover, the largest workflows in their experiments
have up to 1,000 tasks, an indication that the genetic algorithm is quite time-consuming.
A similar rationale of moving tasks to locations with sufficient green energy is used by 
Hanafy et al.~\cite{hanafyCarbonScalerLeveragingCloud2023c}, who
scale the resources assigned to elastic cloud applications in a carbon-aware manner.
Considering a single data center location, in turn, Wiesner et al.~\cite{wiesnerLetWaitAwhile2021e}
investigate how beneficial shifting of execution times to intervals with lower carbon emissions can be.
By evaluating the impact of time constraints, scheduling strategies, and forecast accuracy, 
they find significant potential (under certain conditions) and provide guidance regarding corresponding data center policies.

Altogether, this work is in line with the general trend of minimizing energy consumption and/or  carbon
emissions.  However, to the best of our knowledge,  it is the first to focus on optimizing the scheduling of a given workflow mapping and ordering to benefit from time-varying green energy.

\section{Framework}
\label{sec.framework}

We use a suitable time unit (e.g. seconds, minutes, \dots) and express all parameters as integer multiples  of this unit.

\paragraph{Platform and application. }
The target platform $\mathcal{C}$ is a cluster of $P$ heterogeneous processors $\{p_1,\dots,p_{P}\}$.
The target application consists of a workflow modeled as a Directed Acyclic Graph (DAG) $G=(V, E,\omega,c)$, 
where the vertex set~$V$ represents the set of $n$ tasks $v_1, \ldots, v_n$.  An edge $(v_i, v_j) \in E$ 
represents a precedence constraint 
between tasks $v_i$ and~$v_j$, meaning that task~$v_j$ 
cannot start before task~$v_i$ is completed and its output was communicated to the processor handling task~$v_i$. 

We assume that the mapping is given, as well as the ordering of the tasks and the communication operations 
(i.e., data transfer) on each processor. Therefore, if two tasks are mapped on the same processor with task~$v_i$
planned before task~$v_j$, we add a precedence constraint $(v_i, v_j)$ to~$E$, to ensure that the order 
is respected. 

Given the mapping, the set of communications is represented by $E' \subseteq E$, which contains all edges $(v_i,v_j)\in E$
such that the two tasks are mapped on different processors, in which case data must be communicated
between both processors before $v_j$ can start its execution. However, when the two tasks are on the same processor
 ($(v_i, v_j) \in E \setminus E'$),  task~$v_j$ can start as soon as task~$v_i$ is finished. 

Each task~$v_i \in V$ has a running time $\omega(v_i)$, and each edge $(v_i,v_j)\in E'$ has 
a communication time $c(v_i, v_j)$, which accounts for the amount of data that has to be communicated 
from the processor of task~$v_i$ to the processor of task~$v_j$. 
Computation and communication times can be arbitrary and are given, which allows us to account 
for any heterogeneity in computing speeds and/or link bandwidths across the processors.

\paragraph{Communication-enhanced DAG $G_c$. }
For simplicity, we assume that the cluster employs 
a fully connected, full-duplex communication topology, where each processor can directly communicate with every other processor simultaneously in both directions. We introduce $P(P-1)$ fictional processors
$\{p_{P+1},\dots,p_{P^{2}}\}$, one per
communication link, whose role is to execute all (potential) communications on that link. 
This will clarify how to compute the cost of a schedule. With these additional processors, each communication 
$(v_{i},v_{j}) \in E'$  
becomes a (fictional) task $v_{i,j}$ of length $\omega(v_{i,j}) = c(v_{i},v_{j})$. 
Furthermore, we add dependencies $(v_{i}, v_{i,j})$ and $(v_{i,j}, v_{j})$, each with zero communication cost. 
Since the order of communications is also assumed to be given with the mapping, we add
precedence constraints to express this order if two tasks $v_{i,j}$ and $v_{k,\ell}$ are on the same
communication link (represented by a fictional processor). This is similar to the precedence constraints
added to express the order of computing tasks and we refer to this set of constraints as $E''$. 

We  obtain a communication-enhanced DAG $G_c= (V_c, E_c, \omega)$, where $V_c$  contains
both $V$ and all $|E'|$ communication tasks~$v_{i,j}$:
$$V_c = V \cup \{ v_{i,j} \; | \; (v_i, v_j)\in E'\}, $$    
and $E_c$ contains both the precedence relations
expressing the order on each processor ($E\setminus E'$) and the new dependencies related to communication tasks: 
\[
E_c = (E \setminus E') \cup \{(v_i, v_{i,j}), (v_{i,j}, v_j) \; | \; (v_i, v_j)\in E'\} \cup E''.
\]
This DAG does not have any communication costs, since they have all been replaced by tasks. 
The number of tasks is $N = |V_c| = n + |E'|$, the mapping of tasks on processors is given as well 
as the order of tasks on processors (both original tasks from~$V$ and communication tasks from~$E'$). 
The construction of the extended platform with $P^{2}$ processors and the corresponding DAG~$G_c$ 
is straightforward. 

\paragraph{Power profile. }
In every time unit, processor $p_{i}$, $1 \leq i \leq P^{2}$, consumes idle power of $\idleP{i}$ units, 
to which a working power of $\workP{i}$ units is added whenever $P_{i}$ is active,
for a total power $\PP{i}(t)$ at time~$t$. 
A processor executing a task or a communication is active from that operation's start to its end. 
Note that communication processors are likely to  
consume much less than regular (computing) processors. In particular, we could set the static power of a link that is never used to $0$.

The horizon is an interval $[0,T[$, where $T$ is the deadline. We assume that the horizon is divided into $J$ intervals
$\{I_{1}, \dots, I_{J}\}$, where interval $I_{j}$ has length $\ell_{j}$ and $\sum_{j=1}^{J} \ell_{j} =T$.
We let $I_j = [b_j, e_j[$ so that $\ell_j = e_j -b_j$ for every $1 \leq j \leq J$. 
The set of starting and ending times of the $J$ intervals is 
$$\mathcal{E} = \{b_1=0, e_1=b_2, e_2=b_3, \ldots, e_{J-1}=b_J, e_J=T\}.$$
Within each interval~$I_{j}$, there is a (constant)
green power budget $\budget_j$ for each time unit~$t \in I_{j}$. 
If the power consumed by all processors at time~$t$ exceeds this budget, 
the platform must resort to brown carbonated power, which will incur some carbon cost at time~$t$. 
This is the key hypothesis of this work: the carbon cost of a schedule will depend in the end 
on which intervals are heavily used, or not,
by the processors. The scheduler must maximize the benefit from greener intervals while enforcing 
all dependencies and meeting the global deadline~$T$.

\paragraph{Carbon cost. }
Given a {\em schedule}, i.\,e., a start time for each task of~$V_c$ (i.\,e., including communication tasks), 
it is easy to compute 
its total carbon cost by looping over the $T$ time units: for each time unit $t$,  
sum up the power consumed by each processor, 
either computing or communicating,  $\mathcal{P}_{t} =  \sum_{i=1}^{P^2} \PP{i}(t)$
(which may include $\workP{i}$ or not).  
The carbon cost for $t \in I_{j}$
is assumed to be proportional to the non-green power, and hence we simply write 
$\mathcal{CC}_{t} = \max(\mathcal{P}_{t} - \budget_{j}, 0)$. 
The total carbon cost of the schedule is then $\mathcal{CC} = \sum_{t=0}^{T-1} \mathcal{CC}_{t}$.  
However, this approach has exponential (in fact,
pseudo-polynomial) complexity, since the problem instance has size 
$O \big( P^2+N+J+\log(T)+\max_{1\leq i\leq P^2} \log(\idleP{i}) +\max_{1\leq i\leq P^2} \log(\workP{i}) + 
\max_{1\leq j \leq J} \log(\mathcal{G}_j) \big)$.  
To compute the cost of a schedule in polynomial time, we need to proceed interval by interval 
and create sub-intervals each time a task starts or ends, so that the number of active tasks is constant 
within each subinterval. The carbon cost per subinterval then depends on the power cost of the subinterval
($\mathcal{CC}_{t}$ is constant within a subinterval) and the interval length. 
Details can be found \new{in Appendix~\ref{app.cost}.} 

\paragraph{Optimization problem. }
The objective is to find an optimal schedule, defined as a schedule whose total carbon cost $\mathcal{CC}$ is minimum. 
To achieve this goal, the scheduler can shift around tasks (including communication tasks) 
to benefit from greener intervals, while enforcing all dependencies and meeting the deadline~$T$.

\section{Complexity Results}
\label{sec.complexity}

In this section, we present an involved dynamic programming (DP) algorithm, 
establishing that the problem with a single processor has polynomial time complexity. 
On the contrary, and as expected, the problem with several processors is strongly NP-complete, 
even with homogeneous processors and no communications, but we can formulate
the general problem as an integer linear program (ILP). 

\subsection{Polynomial DP algorithm for one processor}
\label{sec.comp.seq}

\begin{theorem}
\label{th.oneproc}
The problem instance with a single processor has polynomial time complexity.
\end{theorem}

\begin{proof}
Consider the problem instance with a single processor executing tasks $v_{1}, \dots, v_{n}$ in this order
(with $n=|V|$, no communication tasks in this case). 
We start with a pseudo-polynomial algorithm: for $1 \leq i \leq n$ and $1 \leq t \leq T$,
we let $\oopt(i,t)$ be the cost of an optimal schedule for the first $i$ tasks and where task $v_{i}$ completes its execution exactly at time $t$, storing the value $\infty$ if no such schedule exists. We have the induction formula
\begin{equation}
\label{eq.progdynpseudo}
\oopt(i,t) = \min_{s \leq t-\omega(v_i)}  \big\{ \oopt(i-1,s) + \texttt{cc}(v_i, t) \big\}, 
\end{equation}
where $\texttt{cc}(v_i, t)$ 
is the cost to execute task~$v_i$ during the interval
$[t-\omega(v_i),t[$. 
Since there is a single processor, this can be computed in linear time by computing the length of its intersection with the intervals~$I_{j}$. In Eq.~\eqref{eq.progdynpseudo}, we  loop over possible termination dates for the previous task~$v_{i-1}$. 
For the initialization, we simply 
compute the value of
$\oopt(1,t)$ for all $t \geq \omega(v_{1})$, and let $\oopt(1,t) = \infty$ for $t < \omega(v_{1})$.

This dynamic programming algorithm is pseudo-polynomial because it tries all possible values $t \in [1,T]$
for the end times of the tasks. To derive a polynomial-time algorithm, we show that we can derive an optimal algorithm while restricting to a polynomial number of end dates. 

Given any single processor schedule $\mathcal{S}$, we define a \emph{block} as a set of consecutive tasks in the schedule, 
i.\,e., there is no idle time between tasks within a block. 
Note that if a task has idle time before and after it, it forms a block by itself. 
Furthermore, 
schedules where each block either starts or ends at a time in $\mathcal{E}$ are called $\mathcal{E}$-schedules
(recall that $\mathcal{E} = \{b_1=0, e_1=b_2, e_2=b_3, \ldots, e_{J-1}=b_J, e_J=T\}$ is the set of starting/ending
times of intervals). 
We can then prove the following Lemma (see proof \new{in Appendix~\ref{app.comp.blocks}}). 
\begin{lemma}
\label{th.lemmablock}
With a single processor, there always exists an optimal $\mathcal{E}$-schedule.
\end{lemma}
According to this lemma, we can therefore restrict the pseudo-polynomial dynamic programming algorithm to only using
task end times that belong to a refined set of end times~$\mathcal{E'}$, which is of size $O(n^{3} J)$ 
(\new{see Appendix~\ref{app.comp.blocks}}) 
thereby leading to a fully polynomial running time.
\end{proof} 

\subsection{NP-completeness of the multiprocessor case}
\label{sec.comp.para}

\begin{theorem}
\label{th.severalproc}
The problem instance with several processors is strongly NP-complete, even with uniform processors and independent tasks (hence no communications).
\end{theorem}

\begin{proof}
We consider the class \texttt{UCAS} of decision problem instances with $P$ processors 
with uniform power consumption, i.e., $\idleP{i} = 0$, $\workP{i} = 1$ for  $1 \leq i \leq P$, and 
an input DAG $G=(V, E, \omega, c \equiv 0)$. Given a bound $C$, we ask whether there exists a valid schedule 
whose total carbon cost does not exceed $C$. We prove \new{in Appendix~\ref{app.comp.para}} 
 that
  \texttt{UCAS} is strongly NP-complete by reducing the well-known \texttt{$3$-Partition} problem to it.
\end{proof}

\subsection{Integer linear program}
\label{sec:ILP}

We formulate the problem as an integer linear program. Due to space constraints,
we only sketch the derivation and refer \new{to Appendix~\ref{app:ILP-details}} 
for details.
The ILP is written in terms of time units, hence it has a pseudo-polynomial number of variables. As stated in Section~\ref{sec.framework}, the objective function
is then to minimize
\begin{equation}\label{eq:objective}
  \mathcal{CC} = \sum_{t=0}^{T-1} \max\left( \sum_{i=1}^{P^2} \left(\idleP{i} +  \delta(t, i) \workP{i} \right)- \mathcal{G}_t, 0 \right) , 
\end{equation}
where $\delta(t,i)$ is a boolean variable that specifies whether processor~$p_{i}$ is active at time~$t$. 
Note that we still use the communication-enhanced graph, and the 
ILP enforces all dependence constraints and guarantees that all tasks are completed by the deadline $T$.

\section{Algorithms}
\label{sec:algorithms}
\label{sec.algos}


In this section, we present \texttt{CaWoSched}, a carbon-aware workflow scheduler for the scheduling problem
of minimizing the carbon cost, given a mapping and a deadline. 
Recall that we work on the communication-enhanced DAG $G_c= (V_c, E_c, \omega)$. 
In Section~\ref{sec.algo.preprocessing}, we first present the baseline algorithm, 
\EST (As Soon As Possible), which schedules each task at its earliest possible start time,
without taking the intervals into account. 
Section~\ref{sec.algo.greedy} presents several variants of a greedy procedure that allocates start times
to tasks, building on a score that is computed for each task. 
Finally, we explain  in Section~\ref{sec.algo.localSearch} how to further improve the schedule obtained 
by the greedy algorithm, by using local search. 

\subsection{Baseline algorithm} 
\label{sec.algo.preprocessing}
\label{par.algo.EST_LST}
The \EST baseline algorithm starts each task at their earliest possible start time ($EST$). 
To compute these times, we proceed similarly to the computation of a topological ordering. 

For all tasks~$u\in V_c$ with in-degree~$0$ (guaranteed to exist since $G_c$ is acyclic), we set
$EST(u) = 0$ and decrease by one the in-degree of successor tasks (tasks~$v$ such that $(u,v)\in E_c$). 
A task~$v$ obtains an in-degree of~$0$ once all of its predecessors have been handled, and we can 
then compute its earliest start time as:
\begin{equation*}
  EST(v) = \max_{(u,v) \in E_c} \{EST(u) + \omega(u)\},  
\end{equation*} 
which corresponds to the time when all predecessors have completed their execution, 
when they are started as soon as possible.

%
The computation of $EST$ is done with a queue to handle tasks. 
The proof for correctness and existence is similar to the proof of correctness for Kahn's algorithm 
for topological sorting~\cite{Kahn62} and is hence omitted here.


\subsection{Greedy schedule}
\label{sec.algo.greedy}
We now describe how to compute a greedy schedule for the workflow, while accounting for the 
carbon cost of each interval (\EST does not consider intervals at all). 
%
%
The idea is to assign a score to each task, and sort the tasks accordingly. 
Afterwards, we process the tasks in this order and try to find a good starting time for them.

\paragraph{Scores for the tasks.}\label{par.algo.scores}
The goal of the scores is to express how beneficial it is to schedule a task before other tasks. 

The first  score is the {\bf slack} $s(v)$ of task~$v$, which represents 
the difference between the 
latest possible starting time of a task~$v$, $LST(v)$,  
and its earliest start time $EST(v)$.

$LST$ can be computed similarly to $EST$, using a queue to handle tasks. 
We set $LST(v) = T - \omega(v)$ if $v\in V_c$ and decrease by one the out-degree of predecessor tasks
(tasks~$u$ such that $(u,v)\in E_c$). A task~$u$ obtains an out-degree of~$0$ once all of its successors 
have been handled, and we can 
then compute its latest start time as:
\begin{equation*}
  LST(u) = \min_{(u,v) \in E_c} \{LST(v) - \omega(u)\}. 
\end{equation*}

Hence, the slack $s(v) = LST(v) - EST(v)$ describes the number of time units by which a task~$v$ can be shifted, 
since its start time has to be between $EST(v)$ and~$LST(v)$. 
If the slack of a task is large,
then it usually means that we have some flexibility to schedule it. 
We therefore try first to schedule tasks with a small slack, since there will still be room to shift tasks
with a higher slack later. Note, however, that the slack does  
not account for the running time of the task. 

The second score is the {\bf pressure} of a task~$v$, defined as: 
 \begin{equation*}
  \rho(v) = \frac{\omega(v)}{s(v)+\omega(v)}.
\end{equation*} 

While slack does not take the running time into account, pressure accounts
for it since it might play an important role for the power usage of the cluster. 
For pressure values, we have $0 \leq \rho(v) \leq 1$, with a pressure of~$1$ when there
is no flexibility (i.e., $s(v)=0$). 

Hence, on the one hand, there is a high pressure to schedule a task $v$ if its running time is large compared 
to the range in which it can run. 
On the other hand, if a task has low pressure, it means that there is a lot of flexibility for starting the task. 
In this case, it is beneficial to schedule tasks with high pressure first; hence, 
we sort the tasks by non-increasing order of pressure. 

However, both scores do not account for the heterogeneity of the processors in terms of power consumption. 
Hence, we also introduce two weighted scores, where the functions 
for a task~$v$ mapped on processor~$p_i$ are multiplied by the factor: 
\begin{equation*}
wf(i)=  \frac{\idleP{i} + \workP{i}}{\max_{j}(\idleP{j}+\workP{j})}
\end{equation*} for pressure and its reciprocal for slack. 
For slack, we use the reciprocal since tasks are sorted in non-decreasing order.

\paragraph{Subdivision of the intervals.}\label{par.algo.subdivision}
Recall that $I_1,\ldots, I_J$ are the initial intervals coming from the power profile. 
As discussed in Section~\ref{sec.comp.seq}, 
there is a more fine-grained subdivision 
of these intervals, such that every task starts at the beginning of such an interval when we look at the special case 
of one processor. Motivated by this result and the question of how
to find a good starting time for a task without looking at every time unit, 
we do a similar subdivision for the multiprocessor case. 
On each processor, we create all possible blocks of at most~$k=3$ consecutive tasks (the parameter~$k$
is used to limit the number of intervals, and hence the time complexity of the heuristics). Each block
is tentatively scheduled to start or end at one of the original intervals, and we memorize the possible 
start times for each task and each block. When this is done on all processors, 
we sort the possible start times 
and compute the induced subdivision of the intervals. 

Further refinement could be used by considering larger block sizes $k>3$, but we observed in our experiments that $k=3$ already creates 
a lot of \old{sub-intervals} \new{subintervals}.

\paragraph{Algorithm variants without local search.}
With four scores (slack, pressure, weighted-slack, weighted-pressure) and two interval subdivisions
(normal or refined), we obtain eight algorithm variants:
\algvar{slack} (unweighted, normal), \algvar{slackW} (weighted, normal), \algvar{slackR} (unweighted, refined), \algvar{slackWR} 
(weighted, refined) for slack and analogously with prefix \algvar{press} for pressure.

We now detail how these algorithms select a starting time for each task,
which is always a time at the beginning of an interval. 
%

Given a score and an interval subdivision, we pick the next task, say $v$, according to the best score value.
The interval set is denoted as 
$\{I_1, \dots, I_{J'}\}$, where $J'=J$ if intervals are not refined, and \old{$J' > J$} \new{$J'\geq J$} otherwise. 
We have 
 $I_j = [b_j, e_j[$. 
 First, the algorithm computes the subset of the intervals such that 
$EST(v) \leq b_j \leq LST(v)$, i.e., intervals at the beginning of which the task can be started. 

If this set is empty (which is rarely the case in practice), we simply start the task at time~$EST(v)$. Otherwise, 
we sort the intervals according to their budget $\mathcal{G}_j$ and schedule the task to start at the beginning 
of the interval with the highest budget. If there are multiple intervals that are possible, we use the interval with the earliest starting point. 

Afterwards, 
we look at all intervals during which task $v$ runs. For the first and last intervals, if $v$ does not cover the whole interval, 
we split the interval in two sub-intervals (one where the task is running, the other where it is not). 
Then, on each interval where the task runs, we decrease the power budget
by $\idleP{i}+\workP{i}$, where $p_i$ is the processor on which task~$v$ is mapped, to account for the fact
that there is now a task running in this interval and consuming some power -- hence the green budget is lower. 

Also, once the task has been scheduled, this influences the $EST$ and $LST$ of other tasks as well. 
Hence, we update this for all tasks that have not been 
scheduled yet. In particular, these updates have to be made possibly for the whole graph, and we use a precomputed topological order for this. 
These updates take~$O(n+|E_c|)$ time. 

\subsection{Local search}
\label{sec.algo.localSearch}
Once a greedy schedule has been obtained, we propose to refine this schedule by doing a local search, 
exploiting the flexibility that tasks still provide within the greedy schedule. 
The corresponding algorithm variants receive a suffix of \algvar{-LS}, for example \algvar{pressWR-LS}.

For the local search, we introduce a parameter 
$0 \leq \mu \leq T-1$. First, we sort the processors by non-increasing power consumption $\workP{i}$,
i.e., the more costly processor is considered first. 
For each processor in this order, we then iterate over the tasks of the processors from left to right
and look $\tau$ time units to the left and \new{right, and} check whether moving the task would give us a gain and is valid. 
This means that we make sure for every possible move that the corresponding start time of a task $v$ stays in the interval $[EST(v), LST(v)]$.
\old{Thereby, we always store the best possible legal move and after checking all
specified time units for the task, we apply the best move if it has positive gain and we update the cost.}
\new{To this end, we iterate over the time units from the earliest to the latest. If we find a legal move with a positive gain, we apply it and update the cost.
(One could also check all legal moves and apply the best one. However, preliminary experiments showed that this would not significantly improve the outcome, 
so we opted for the faster variant.)}
Once this has been done for all tasks on the current processor, we process the tasks
of the next processor in the ranking. 
At each round, we record whether we had a positive gain or not. If there was one round through the tasks without gain, 
we stop the local search. 

\section{Experimental Evaluation}
\label{sec.expes}
In this section, we evaluate the proposed carbon-aware scheduling framework \texttt{CaWoSched} with its numerous algorithm variants.
We mainly focus on solution quality and running time in comparison to the baseline \EST. For small instances, we also
compare the quality against optimal solutions derived from the ILP formulation.
The code and data used in the simulations are publicly available for reproducibility purposes at \new{\url{https://github.com/KIT-EAE/CaWoSched}}.

\subsection{Simulation setup}\label{sec.expe_setup}
\paragraph{Target computing platform.}\label{par.hardware}
%
We consider two target computing platforms with a heterogeneous setup; their
properties are inspired by
real-world machines used for the experimental evaluation in~\cite{bader2024lotaru}. 
There are six processor types, and we consider $12$ (resp.~$24$) nodes of each type for the
\cluster{small} (resp. \cluster{large}) cluster, 
hence a total of  $72$ (resp. $144$) compute nodes. 

In addition to the normalized speed values (see Table~\ref{table:clusterSetup}),   
we assign each processor a value for its idle power consumption $\idleP{}$ and its active power consumption $\workP{}$. 
The values for the power consumption are inspired by 
values coming from Intel~\cite{TDPvsACP} for modern processors. Note that we did not choose 
$\workP{}$ too large since  the CPU utilization in data centers is often far from $100\%$~\cite{googleDataCenter}. 
While the correlation between power consumption and processor speeds may not be obvious, 
the general trend is that faster processors consume more power, hence a ranking of the processor types:
nodes of type $PT1$ are the slowest/least consuming nodes, up to $PT6$, which are the fastest/most consuming ones. 
According to~\cite{googleDataCenter}, the power consumption of the network is much smaller than that of computation,
hence we draw the values for $\idleP{}$ and $\workP{}$ randomly between $1$ and $2$ for communication links,
in order to introduce a small amount of heterogeneity. 
%

\begin{table}[b]
  \centering
  \caption{Processor specifications in the clusters.}
  \label{table:clusterSetup}
  \begin{tabular}{c c c c c c}
    \toprule
    Processor Name & Speed & $\idleP{}$ & $\workP{}$ & \texttt{small} & \texttt{large} \\
    \midrule
    $PT1$ & 4 & 40 & 10 & $\times 12$ & $\times 24$ \\
    $PT2$ & 6 & 60 & 30& $\times 12$ & $\times 24$ \\
    $PT3$ & 8 & 80 & 40 & $\times 12$ & $\times 24$ \\
    $PT4$ & 12 & 120 & 50 & $\times 12$ & $\times 24$ \\
    $PT5$ & 16 & 150 & 70 & $\times 12$ & $\times 24$ \\
    $PT6$ & 32 & 200 & 100 & $\times 12$ & $\times 24$ \\
    \bottomrule
  \end{tabular}
\end{table}

\paragraph{Workflows and mappings.}
\label{par.graphs}
We evaluate the presented algorithmic framework on \new{$34$} workflows. The corresponding DAGs can be divided into real-world workflows obtained from~\cite{bader2024lotaru} (atacseq, bacass, eager and methylseq)
and workflows obtained by simulating real-world instances using the WFGen generator~\cite{WfCommons}. 
We transformed the corresponding definition for the workflow management system Nextflow~\cite{nextflow} to a .dot format with a Nextflow tool. 
Since the resulting DAGs contain many pseudo-tasks that are only internally relevant for Nextflow, we deleted them,
following what was done in~\cite{kulagina2025}.
For the simulated workflows, we use one of the respective real-world instances as a model graph and scale it up in size.
As number of vertices, we use \numprint{200}, \numprint{1000}, \numprint{2000}, \numprint{4000}, \numprint{8000}, \numprint{10000}, \numprint{15000}, \numprint{18000}, \numprint{20000}, \numprint{25000} and \numprint{30000}.  
Every graph has vertex and edge weights following a normal distribution, where we make sure that the vertex weights are in general larger than the edge weights. 
Note that these are normalized values, and the actual running time of the task is determined by its vertex weight and its assigned processor. 
We normalize the network communication bandwidth to~$1$ since its influence is not considered here.

Furthermore, we generate for every graph two mappings, one for cluster \cluster{small} and one for cluster \cluster{large}.
The mappings are generated with our own basic HEFT implementation without special 
techniques for tie-breaking, because that would not change the fact that HEFT is not carbon-aware.
Since there are more fast and power-intensive processors on the  \cluster{large} cluster, 
HEFT schedules more tasks to these processors and hence there are fewer tasks per processor 
on the other processors, compared to the \cluster{small} cluster. 


\paragraph{Power profiles.}
\label{par.powerProfiles}
For each workflow, we generate four differently shaped (renewable) energy profiles for different scenarios. 
We make sure that green power is always at least the sum of the idle power of the processors and at most 
the sum of idle power and $80\%$ of the sum of the work power. The rationale is as follows: 
if we do not have enough green power or more green power than required overall, the decisions of the scheduler become irrelevant. 
Hence, we try to create scenarios where scheduling decisions have to be done in a smart way. The scenarios are the following:
\begin{enumerate}
  \item[S1:] A $-x^2$ shape, where the interval budgets follow this function with random perturbations. 
    This models a situation where there is little green power in the beginning, then the supply with green energy is rising and falls at some point again (solar power from morning to evening, for example).
  \item[S2:] An $x^2$ shape that models the same situation as in S1, but starting from midday, again with random perturbations.
  \item[S3:] A $\sin(x)$ shape, where we model 24 hours of this scenario, i.e., we have little green power in the beginning and then we follow a sinus shape as given on $[0,2\pi]$. We also add random perturbations.
  \item[S4:] A constant green power budget with perturbations (which can model situations where one has storage for renewable energy or nuclear power -- see setting of France in~\cite{wiesnerLetWaitAwhile2021e}). 
\end{enumerate} 
For each scenario, we have four different deadlines. Let $D$ be the time required by the \EST schedule, 
which is the tightest deadline. We consider deadlines $D$, $1.5D$, $2D$, and $3D$, providing
more or less flexibility to shift tasks around in the schedule. 
 
Hence, we have in total $16$ power profiles. For the workflow types atacseq \new{and} methylseq, we have 12 graphs per type, for bacass we only use the real-world version 
due to problems with scaling, and for eager we have 9 graphs with up to \numprint{18000} vertices. This results in $2 \times 34 \times 16 = 1088$ simulations 
(2 platforms, 34 workflows, 16 power profiles) per algorithm. 
All algorithms are implemented in C++ and compiled with g++ (v.13.2.0) with compiler flag -O3.
The experiments are managed by simexpal~\cite{DBLP:journals/algorithms/AngrimanGLMNPT19} and executed on workstations with 192 GB RAM and 2x 12-Core Intel Xeon 6126 @3.2 GHz
and CentOS 8 as OS.
Code, input data, and experiment scripts are available to allow reproducibility of the results \new{at~\url{https://github.com/KIT-EAE/CaWoSched}}.
\old{We plan to make them publicly available upon paper acceptance.}
The ILP is implemented using Gurobi's Python API~\cite{gurobi} and for license reasons executed on a machine with a 13th Gen Intel(R) Core(TM) i7-1355U processor with 16GB RAM running 
Ubuntu 24.04.1 LTS.
Further, for the simulation results below, we set the tuning parameter for the subdivision \new{to} $k=3$ and the tuning paramater for the local search to $\mu=10$.

\subsection{Simulation results}
\label{sec.expe_quality}
%
We compare the quality of the schedules returned by the \algvar{CaWoSched} variants.
Recall that there are two base scores, slack and pressure, that can be weighted by a factor
accounting for the heterogeneity in power consumption of the processors, and we can use either
the original or refined intervals (see 
Section~\ref{sec:algorithms}). The heuristics then apply a local search to further improve the solution. 
We first compare the solution quality when the local search is applied, but we also analyze 
the influence of the local search 
on different algorithm variants. Next, we study the impact of various parameters. 
We also compare the \new{heuristics'} solution to the optimal solution returned by the ILP on small instances. 
%

\paragraph{With local search.}
As a first measure of performance, we rank the different algorithm variants for each instance. This means that we record the frequency with which each algorithm variant ranks first, second, 
third, etc., in terms of carbon cost. Note that this means that if two variants have the same cost, 
they will end up with the same rank and the next rank is then skipped. The results can be 
seen in Figure~\ref{fig.ranking_distribution}; the main observations are the following. 
\begin{figure}
  \centering
  \includegraphics[width=\linewidth]{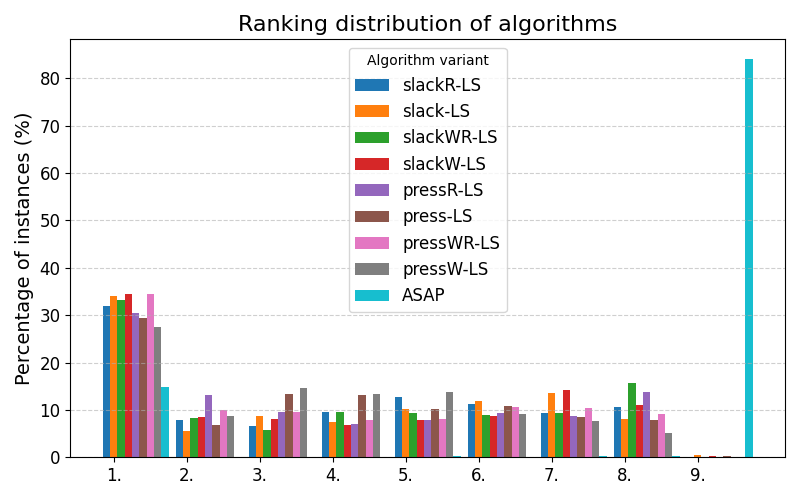}
  \caption{The distribution shows for which percentage of the instances each algorithm variant was ranked first, second, third, etc. Note that multiple algorithm variants can have the same rank.}
  \label{fig.ranking_distribution}
\end{figure}
First, we can see that all our algorithm variants are ranked first significantly more often than this is the case for the baseline \EST. Note that even 
if the baseline has rank~$1$, it does not necessarily mean that a variant of the algorithm performed worse,
since they all could have found the optimal solution. 
In particular, we can see that the baseline performed worst in $\numprint{84.01}\%$ of the cases.
Another observation is that none of the algorithm variants does significantly outperform the others in
terms of ranking. The \algvar{pressWR-LS} variant is ranked first most frequently ($\numprint{34.47}\%$), but with a small margin. 

If we look at performance profiles, we obtain more detailed insights about the
quality of each variant, see Figure~\ref{fig.performance_profile}.
\begin{figure}
  \centering
  \includegraphics[width=\linewidth]{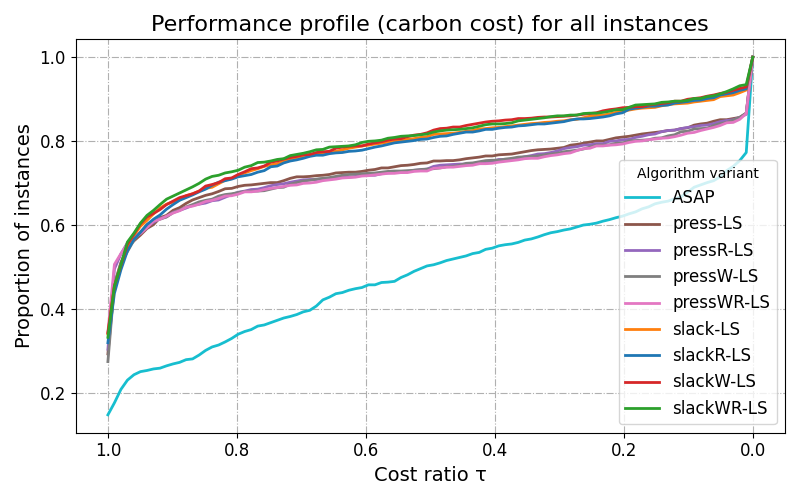}
  \caption{The ratio is the best cost found divided by the algorithm variants' own cost. Then the percentage of instances for which this is larger than or equal to $\tau$ is plotted. A higher curve is better.}
  \label{fig.performance_profile}
\end{figure}
We report the proportion of instances with a cost ratio at most~$1$, where the cost ratio is the best carbon cost divided 
by the algorithm's carbon cost. Note that if the algorithm's carbon cost is 0, then the best cost is also~0 and the
ratio is set to~$1$. Otherwise, a ratio of $1/2$ means that the heuristic's cost is twice higher than the best cost, 
and a ratio of~$0$ corresponds to a non-null carbon cost, while the best is~$0$. 
Even though we can see here again that for $\tau=1.0$, i.e., the proportion of instances for which the algorithm variant achieves the best cost, is the highest for \algvar{pressWR-LS}, we also 
observe that for lower $\tau$ values, i.e. more suboptimality tolerance, the curves for the algorithm variants using slack as base score surpass the pressure variant, which hints at a better overall performance on average.
Interestingly, this observation seems to be influenced by the tolerance in the deadline. 
\old{That is why we show how the deadline influences the performance profile next, see Figure~\ref{fig.pf_offset}.}
\new{This is why we next illustrate how the deadline impacts the performance profile, as shown in Figure~\ref{fig.pf_offset}.}
\begin{figure*}
  \centering
  \begin{subfigure}[b]{0.7\textwidth}
    \includegraphics[width=\linewidth]{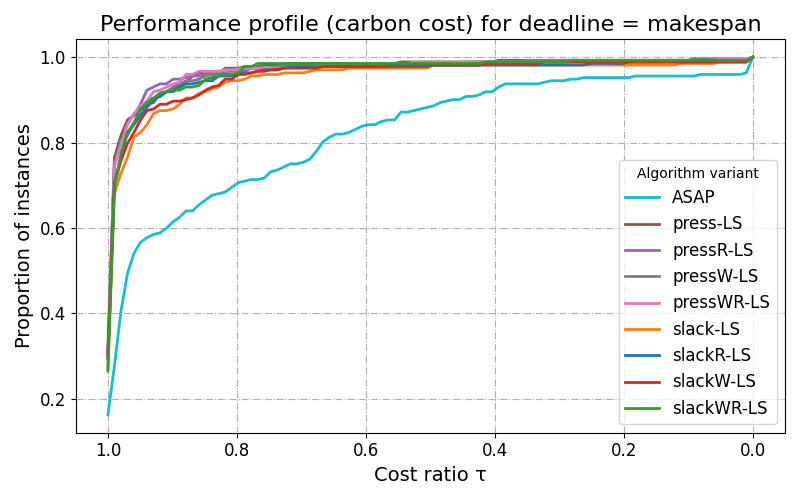}    
  \end{subfigure}
  \hfill 
  \begin{subfigure}[b]{0.7\textwidth}
    \includegraphics[width=\linewidth]{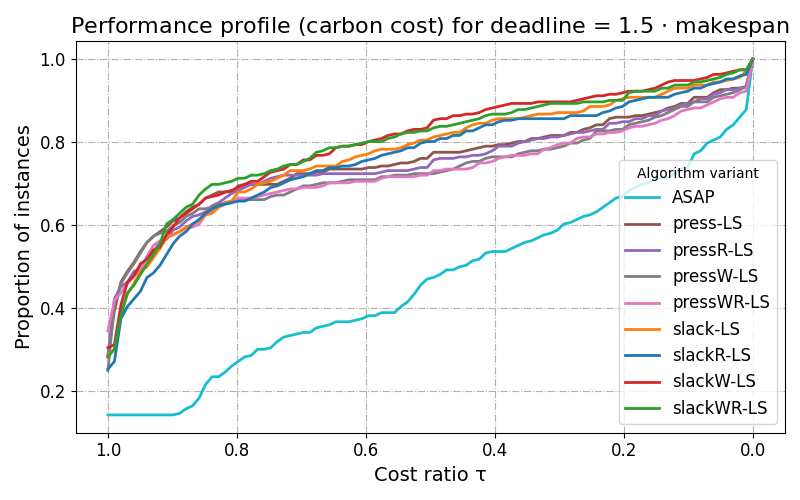}    
  \end{subfigure}
  \hfill
  \begin{subfigure}[b]{0.7\textwidth}
    \includegraphics[width=\linewidth]{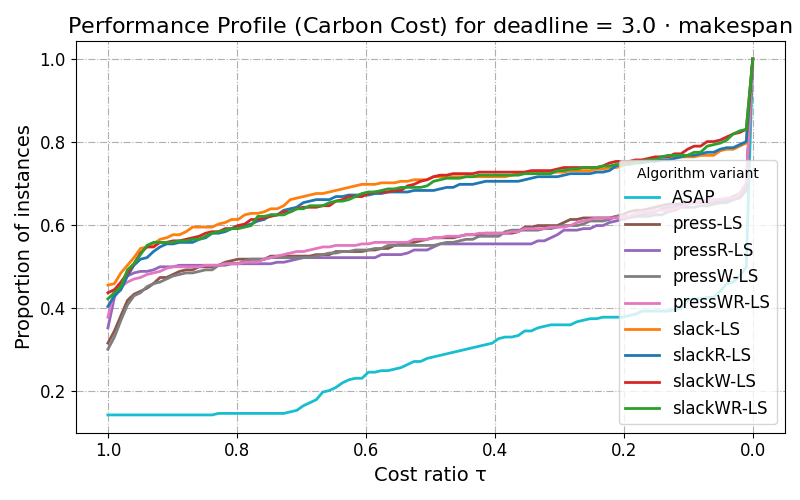}    
  \end{subfigure}
  \caption{The evolution of the performance profiles when adding more tolerance to the deadline from left to right
  (data for deadline factor $2.0$ can be found in Appendix~\ref{sec:moreResults}).}
  \label{fig.pf_offset}
\end{figure*}
While we can observe for a tight deadline that \algvar{pressR} and \algvar{pressWR} have a higher curve, 
these variants are clearly surpassed by slack variants when there is more tolerance in the deadline. 

Another important aspect for the evaluation is the cost improvement over the baseline algorithm \algvar{ASAP}. 
For this, we first look at the median of the cost ratio between the 
baseline \algvar{ASAP} and the different algorithm variants over all instances. This is shown in Figure~\ref{fig.geometric_means}.
\begin{figure}
  \centering
  \includegraphics[width=\linewidth]{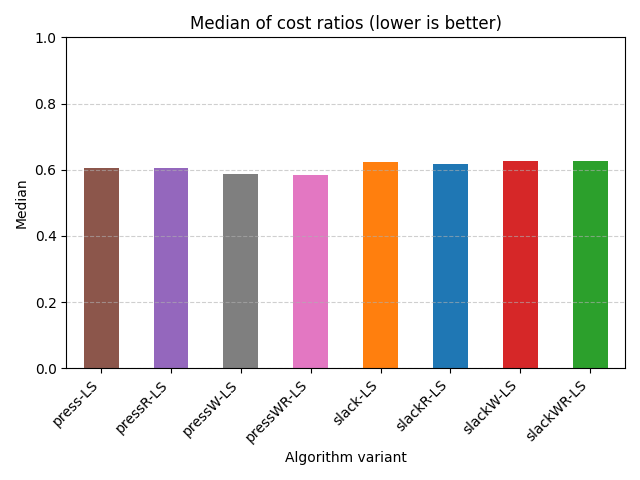}
  \caption{The median of the cost ratios obtained by dividing heuristics carbon cost by the carbon cost of the deadline.}
  \label{fig.geometric_means}
\end{figure}
(Note that a geometric mean is not applicable here because the ratio can be $0$ if our heuristic has carbon cost zero but the baseline has not. 
Further, since there are cases when the baseline performs better than our heuristics, we cannot use the arithmetic mean, either.)
We can see that all algorithms are closely together with a cost ratio median of $\approx 0.6$, meaning that the algorithm needs only $60\%$ of the carbon cost compared to the baseline (or, vice versa, they are $\approx 1.67\times$ better). 
We also see that regarding this cost ratio, the algorithm versions using pressure as base score perform better than the slack variants --
\algvar{pressWR-LS} has the best cost ratio median with $\numprint{0.58}$.
Again, we observe the impact of exploiting more flexibility in terms of deadline in Figure~\ref{fig.means_offset},
where it becomes clear that the cost ratio improves with more flexibility.  
%
\begin{figure*}
  \centering
  \begin{subfigure}[b]{0.6\textwidth}
    \includegraphics[width=\linewidth]{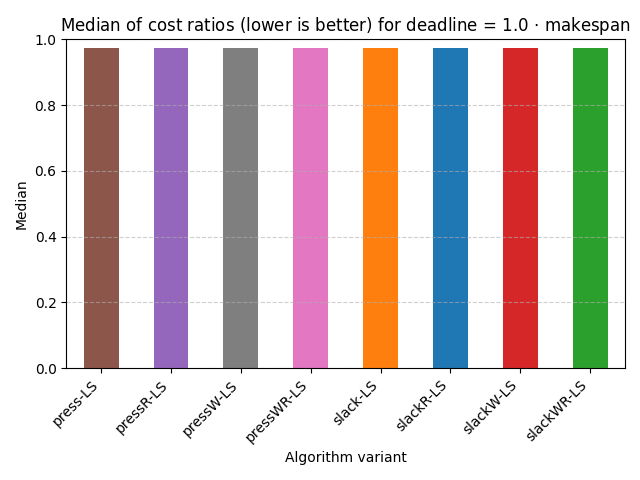}    
  \end{subfigure}
  \hfill 
  \begin{subfigure}[b]{0.6\textwidth}
    \includegraphics[width=\linewidth]{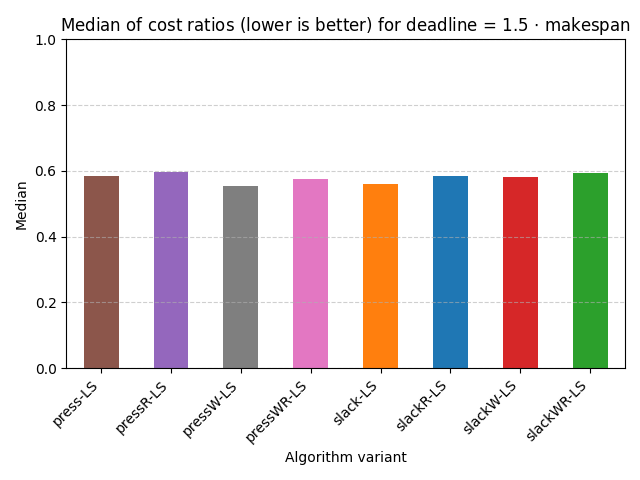}    
  \end{subfigure}
  \hfill 
  \begin{subfigure}[b]{0.6\textwidth}
    \includegraphics[width=\linewidth]{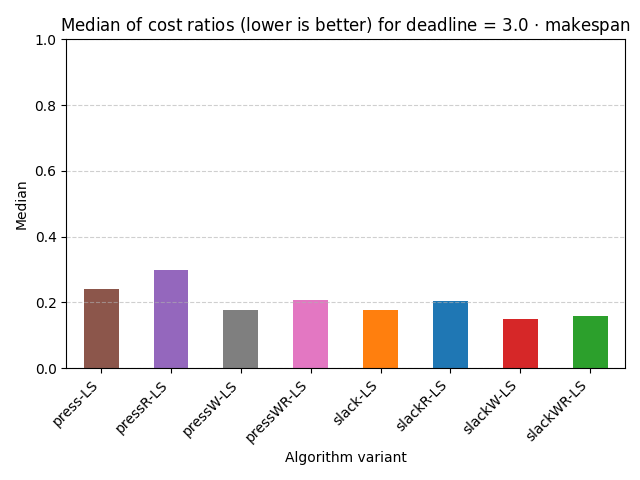}    
  \end{subfigure}
  \caption{The evolution of the median of the cost ratios when adding more tolerance to the deadline from left to right
    (data for deadline factor $2.0$ can be found in Appendix~\ref{sec:moreResults}).
  }    
  \label{fig.means_offset}
\end{figure*}
There, we can see that while the gains for a tight deadline are moderate, the algorithm \algvar{slackW} has a cost ratio of only $0.15$ compared to the baseline (or, vice versa, $\approx 6.67\times$ better).
This is a behavior that we expected from the algorithms, since we have more opportunities 
for scheduling the tasks with an increased deadline. 

To further investigate the 
improvement over the baseline, we look at boxplots for the improvement. The results for all instances are shown in Figure~\ref{fig.costRatio}.
\begin{figure}
  \centering
  \includegraphics[width=\linewidth]{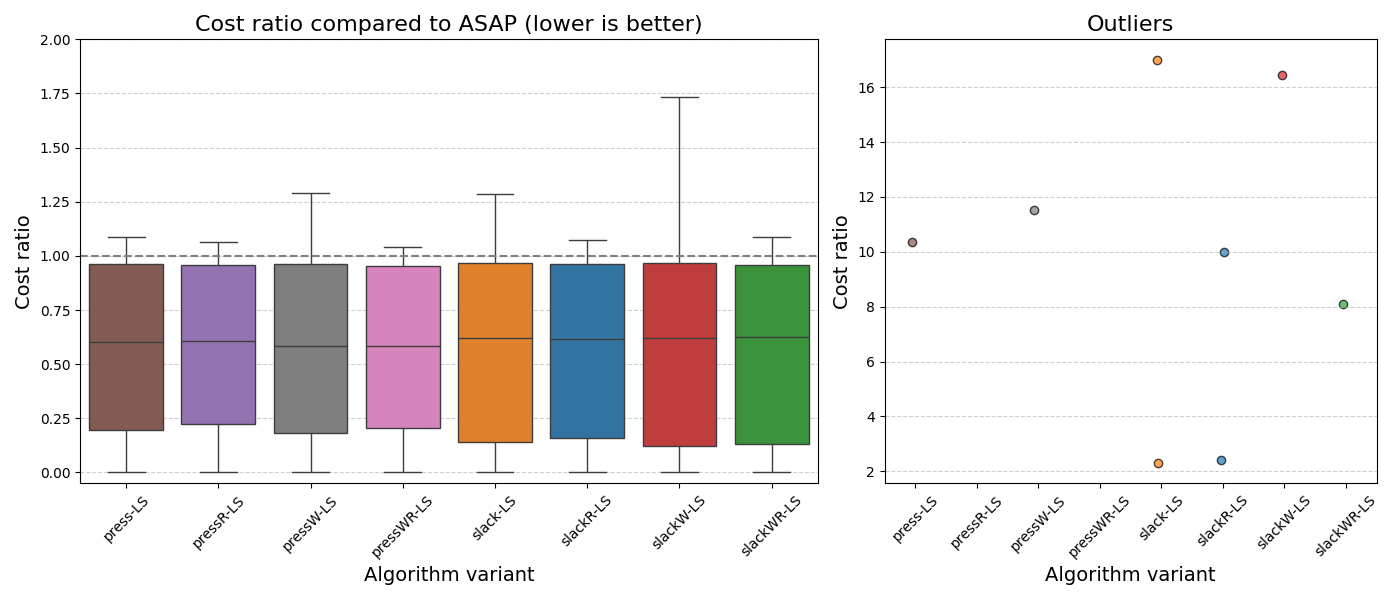}
  \caption{Boxplot of the cost ratios obtained by dividing the heuristics carbon cost by the carbon cost of the baseline. Outliers are shown in the separate plot on the right.}
  \label{fig.costRatio}
\end{figure}
What we can see here is that the solutions for most instances lie in between $\approx 0.25$ and $\approx 0.9$,
with most medians around $0.6$ (compared to the baseline).
We also see that, for some instances, the baseline performs better than the proposed algorithms. 
One reason for this is that some power profiles provide a lot of green power in the beginning 
of the horizon $[0, T[$. Hence, for these profiles, scheduling the tasks as soon as possible 
might be the best strategy. However, overall one can observe that this is rarely the case and that the 
new algorithms significantly improve over \EST in terms of carbon cost.

\paragraph{Influence of local search.}
While we have studied so far the heuristics' behavior when the local search was applied, 
we also run four of the heuristics without the local search to assess how much gain, in terms
of carbon cost, can be achieved thanks to the local search. Note that we use a subset of the full test set 
for this approach, namely all variants of the atacseq workflow type and the bacass workflow. However, note that this still 
results in more than $400$ experiments \new{per algorithm variant}. 
We report the minimum, maximum and average improvement in Table~\ref{table.improvement}. Note that here we use the 
arithmetic mean since the geometric mean is not applicable because we can have a cost ratio of $0$. 
(The results are still meaningful since we only have values between $0$ and $1$.)
\begin{table}
  \centering
  \caption{Minimum, maximum and average cost ratio for comparing the algorithm with and without local search.}
  \label{table.improvement}
  \begin{tabular}{c c c c}
    \toprule
    Algorithm Variant & Min & Max & Avg \\
    \midrule
    \algvar{slackR} & 0 & \numprint{1.0} & \numprint{0.25} \\
    \algvar{slackWR} & 0 & \numprint{1.0} & \numprint{0.25} \\
    \algvar{pressR} & 0 & \numprint{1.0} & \numprint{0.25} \\
    \algvar{pressWR} & 0 & \numprint{1.0} & \numprint{0.23} \\
    \bottomrule
  \end{tabular}
\end{table}
We can see that for every variant the cost ratio ranges from $0$ to $1.0$. A cost ratio of $0$ comes from instances where the algorithm achieves zero carbon
cost using local search, but the algorithm variant without local search has positive carbon cost. Note that a cost ratio larger than $1.0$ is not 
possible since the local search approach is designed as a hill climber. 
In general, we can see that our local search approach significantly improves the solution of the initial schedule \new{(on average up to $\approx 4.35\times$ better)}. In particular, we can observe 
in our experiments that there is a significant number of instances where the local search approach reaches an optimal solution of $0$, but the 
initial schedule has still positive carbon cost. Further, we can see that the local search improves all algorithm variants by a similar margin.

\paragraph*{Impact of parameters.} 
A complete study highlighting the impact of all parameters through detailed results
is available \new{in Appendix~\ref{sec:moreResults}}. 
First, as expected, the \EST baseline performs better when there is a lot of green power
at the beginning of the time horizon or when there is no huge change as in Scenarios S4 or S2. We provide detailed results for each power profile,
while we have presented so far aggregated results. 
Further, we can see that our algorithm achieves a significantly better cost ratio when there is not much green power in the beginning such as in Scenarios S1 and S3.


%
Also, we look separately at the cost ratios depending on the cluster size. We can see here that the cluster \new{size} has no significant influence on the performance of our heuristics. 
However, it influences the performance profile. 
While for the large cluster, the curves are closer together, we see a similar situation as in Figure~\ref{fig.performance_profile}
for the smaller cluster. 

Finally, we also study the impact of the number of tasks on the solution. The general trend is that the cost ratio gets slightly worse when the number of tasks increases. However, 
the effect is not significant, and we can conclude that the improvement of our heuristic over the baseline is in all cases significant. 


\paragraph*{Comparison with optimal solutions.}
\label{sec.expe_ilp}
We explore the quality of the novel heuristics when compared to an exact solution. 
For this, we use the ILP formulation presented in Section~\ref{sec:ILP}, and use the ILP solver Gurobi~\cite{gurobi}. 
Further, our implementation in Python makes use of the NetworkX library~\cite{networkx}. 
We restrict to instances with up to $200$ tasks, since the solver takes too long on larger instances
(already up to one hour for $200$ tasks vs only milliseconds for each heuristic). 
%
Note that the scope of this work is not to explore an efficient ILP formulation for this algorithm, 
we solely want to explore the quality of the novel algorithms. This is why we keep a simple but correct
ILP with time units instead of moving to an interval formulation. 

We show an analysis of the results in Figure~\ref{fig.ILP_comparison}. 
\begin{figure}
  \centering
  \includegraphics[width=\linewidth]{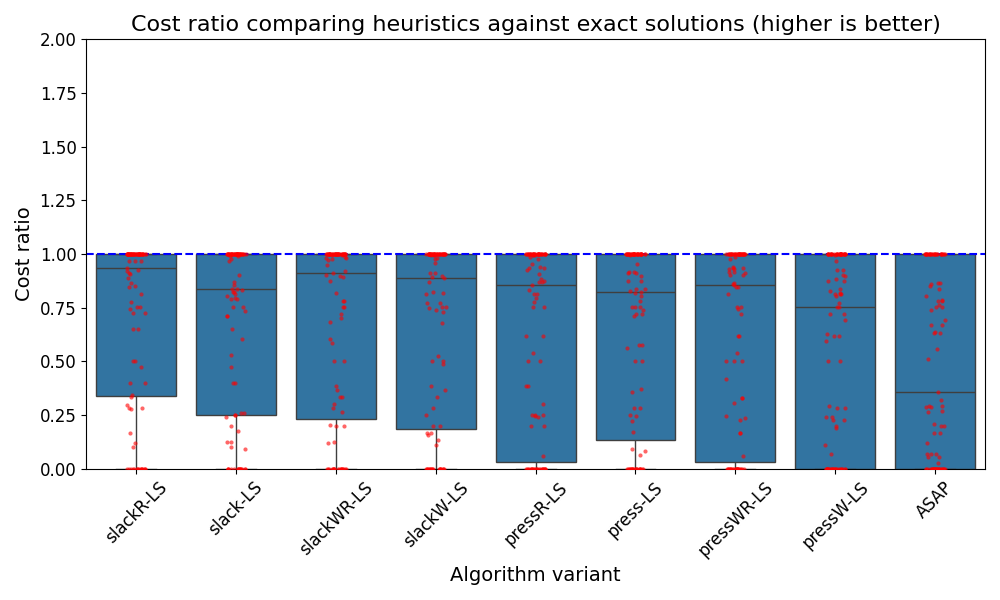}
  \caption{Cost ratio obtained by dividing the ILP result through the heuristic result. Red dots show the actual cost ratios.}
  \label{fig.ILP_comparison}
\end{figure}
What we can see here is that the median cost ratio is still reasonable when we compare our heuristics to exact solutions. Further, this seems to be an achievement of our heuristics since 
we can see that the baseline has a much worse cost ratio than our heuristics. Further, it is interesting to see that there are a significant number of instances where the cost ratio is $1.0$, 
meaning that our heuristic is able to achieve the optimal solution.

\subsection{Running Time Evaluation}
\label{sec.expe_time}
In this section, we present the time needed by the various algorithm variants for computing a carbon-aware schedule. 
Even though it is not the main goal of the scheduler to be as fast as possible, it is important that it is not  
overly time-consuming to be applicable in real-world scenarios.

\begin{figure}
  \centering
  \includegraphics[width=\linewidth]{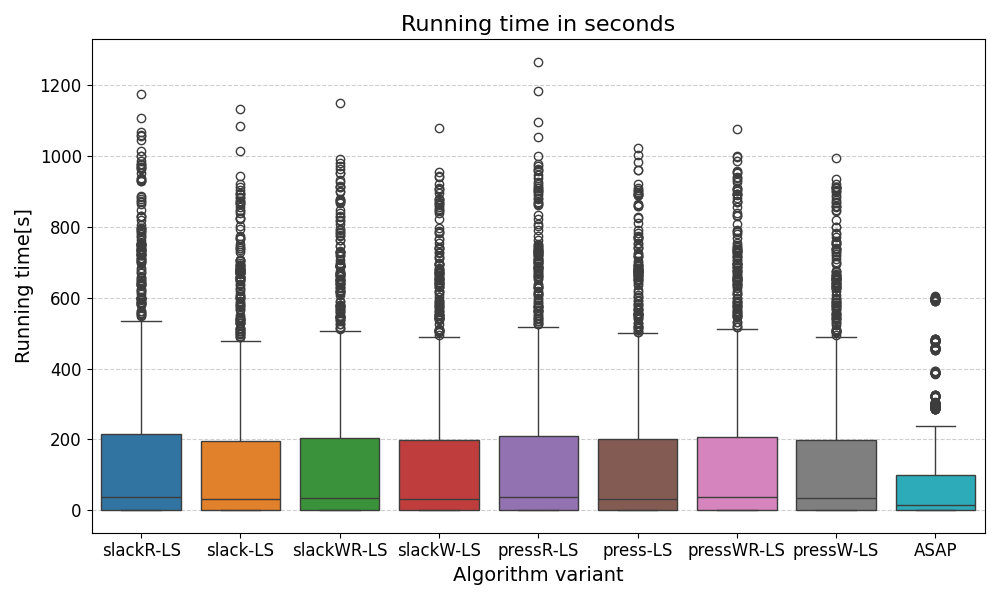}
  \caption{Time (in seconds) for each algorithm variant.}
  \label{fig.time_comparison}
\end{figure}
Aggregated running time values based on all workflows are shown in Figure~\ref{fig.time_comparison}.
We can observe that all algorithm variants  yield a reasonable slowdown compared to the baseline. 
For most of the instances, the scheduler is able to compute a schedule within seconds,
while larger workflows with up to $\numprint{30000}$ tasks can take several minutes.

\old{Figure~\ref{fig.time_comparison_largeWF} in\iftoggle{LV}{ Appendix~\ref{sec:moreResults}}{~\cite[Appendix~\ref{sec:moreResults}]{zenodo}} indicate that the largest running times result from large workflows 
with $\numprint{20000}$ to $\numprint{30000}$ tasks.}
\new{In Appendix~\ref{sec:moreResults} we can see that the largest running times result from large workflows 
with $\numprint{20000}$ to $\numprint{30000}$ tasks.}
Overall, most of the instances are still solvable in \new{less than a few} minutes. 

Another interesting aspect is that the running time seems to be mainly influenced by graph size, but not by the length 
of the time horizon~$T$. 
The running time increases only slightly with an increased deadline, which indicates that the algorithms successfully 
make decisions based on structural graph information -- without having a too broad search tree for each task 
if the deadline increases. 
The data leading to this insight can be found in Appendix~\ref{sec:moreResults} 

\section{Conclusions}
\label{sec.conclusion}

This work aimed  at minimizing carbon emissions when executing a scientific workflow on a parallel platform with a time-varying \old{mix} \new{mixed} (renewable and non-renewable) energy supply.  We focused on improving a given
mapping and ordering of the tasks (for instance generated by HEFT) by shifting task executions 
to greener time intervals whenever possible, while still enforcing all dependencies.
We showed that this algorithmic problem can be solved in polynomial time in the uniprocessor case. For two processors, the problem becomes NP-hard, even for a simple instance with independent tasks and carbon-homogeneous processors.
We proposed a heuristic framework combining several greedy approaches with local search. The experimental results showed that our heuristics provide significant savings in carbon emissions compared to the baseline. 
Furthermore, for smaller problem instances, we showed that several heuristics achieve a performance close to the optimal ILP solution. 
Altogether, all these results represent a major advance in the understanding of the problem.

Future work will be devoted to the next step, namely targeting the design of a carbon-aware extension of HEFT.
Mapping and scheduling the workflow at the same time while minimizing carbon emissions may well lead to even better solutions.
Given the difficulty of the problem, we envision a two-pass approach: a first pass devoted to mapping and ordering, 
but without a finalized schedule, and a second pass devoted to optimizing the schedule through the approach followed in this paper.

\iftoggle{TR}{
\paragraph*{Acknowledgements.}
This work is partially supported by Collaborative Research Center (CRC) 1404 FONDA -- 
\emph{Foundations of Workflows for Large-Scale Scientific Data Analysis}, which is funded
by German Research Foundation (DFG).
}
{}

\bibliographystyle{abbrv} 
\bibliography{sample-base,myChapter,references}

\setcounter{section}{0} 
\renewcommand{\thesection}{\Alph{section}}
\iftoggle{LV}{
\newpage
\section{Appendix}
\label{app.app}

\subsection{Cost of a Schedule}
\label{app.cost}

In this section, we detail how to compute the cost of a schedule in polynomial time.
The schedule~$\sigma$ gives the starting time of each task: task~$u\in V_c$ starts
at time~$\sigma(u)$, and hence completes at time $\sigma(u) + \omega(u)$. 

We proceed interval by interval. 
Recall that we are given $J$ intervals
$\{I_{1}, \dots, I_{J}\}$, where interval $I_{j}$ has length~$\ell_{j}$, and $\sum_{j=1}^{J} \ell_{j} =T$.
We have
$I_j = [b_j, e_j[$ so that $\ell_j = e_j -b_j$ for every $1 \leq j \leq J$, and 
the set of starting and ending times of the $J$ intervals is 
$$\mathcal{E} = \{b_1=0, e_1=b_2, e_2=b_3, \ldots, e_{J-1}=b_J, e_J=T\}.$$

For each interval $I_{j} = [b_{j}, e_{j}[$:
\begin{itemize}
\item We compute the intersection of the execution of each task 
with~$I_{j}$. 
Since the schedule~$\sigma$ is given, we can complete this step in linear time. 
We let $\mathcal{Q}^{j}$ be the subset of tasks intersecting with~$I_{j}$, i.e., 
$$u\in \mathcal{Q}^{j} \quad \Leftrightarrow \quad [\sigma(u), \sigma(u)+ \omega(u) [ \;\;  \cap  \;\; I_j \neq \emptyset . $$
\item For each task $u \in \mathcal{Q}^{j}$, 
we let $[\mathit{start}_{u}^{j}, \mathit{end}_{u}^{j}[ \subseteq [b_{j}, e_{j}[$ denote its execution interval within $I_{j}$. 
\item We sort the list $\{ \mathit{start}_{u}^{j}, \mathit{end}_{u}^{j}\}$ for $u \in \mathcal{Q}^{j}$,
adding $b_{j}$ and~$e_{j}$ if these values are not already in the list, and removing duplicates. 
Let $q_{j}$ be the number of elements of the list. We derive the ordered list
$\mathcal{E}^{j} := \{d_1=b_{j}, d_2, \dots, d_{q_{j}}=e_{j}\}$
and record the length $\lambda_{k}$ of each subinterval $[d_{k}, d_{k+1}[$ for $1 \leq k < q_{j}$.
\item \old{The total power consumed during interval $I_{j}$ is then computed as 
$\mathcal{P}_{j} = \sum_{k=1}^{q_{j}-1} \sum_{i=1}^{P^{2}} \lambda_k \mathcal{P}^{i}(d_k) $,
with the carbon cost  $\mathcal{CC}_{j} = \max(\mathcal{P}_{j} - \budget_{j}, 0)$.
Indeed, the consumed power is constant in each subinterval $[d_{k},d_{k+1}[$, hence we multiply 
the power at the beginning of the interval (time $t=d_k$) by the interval length~$\lambda_k$.}
\item \new{Now, the consumed power is constant in each subinterval $[d_{k},d_{k+1}[$, equal to $\mathcal{P}^{i}(d_k)$ on processor~$i$,
for a green budget of $\budget_j$ in this subinterval. Hence, the carbon cost for this subinterval is the total power above the green budget, 
which is $\max\left(\sum_{i=1}^{P^2} \mathcal{P}^{i}(d_k) - \budget_j, 0\right)$, multiplied by the  
interval length~$\lambda_k$.
Finally, we can compute the total carbon cost of the interval $I_j$ as \begin{equation*}
\mathcal{CC}_j = \sum_{k=1}^{q_j-1} \lambda_k \max\left(\left(\sum_{i=1}^{P^2} \mathcal{P}^{i}(d_k)\right) - \budget_j, 0\right).
\end{equation*}  
} 
%
\end{itemize}
The total carbon cost of the schedule is $\mathcal{CC} = \sum_{j=1}^{J} \mathcal{CC}_{j}$, 
it is obtained in polynomial time using this approach.

\subsection{Proof of Lemma~\ref{th.lemmablock}}
\label{app.comp.blocks}

In this section, we first prove Lemma~\ref{th.lemmablock} and show that with a single processor
there always exists 
an optimal $\mathcal{E}$-schedule. Then, we show that the dynamic programming algorithm 
can restrict to a polynomial number of task end times.

\paragraph{Proof of Lemma~\ref{th.lemmablock}.}
Recall that  the set of starting and ending times of the $J$~intervals is: 
$$\mathcal{E} = \{b_1=0, e_1=b_2, e_2=b_3, \ldots, e_{J-1}=b_J, e_J=T\}.$$

Consider an optimal schedule $\mathcal{S}$ with a single processor. If all blocks are aligned to the beginning or end of an interval,
we already have an $\mathcal{E}$-schedule. Otherwise, there is a block of tasks 
$$v_{r} \rightarrow v_{r+1} \rightarrow \dots \rightarrow v_{s}$$
(we can have $r=s$) that is not aligned: $v_{r}$ starts at time $b_{k}+\alpha$, i.e., 
$\alpha$ units of time after the beginning of interval~$I_{k}$,  and $v_{s}$ ends at time $b_{\ell}+\beta$, 
i.e., $\beta$ units of time after the beginning of interval $I_{\ell}$, where $k \leq \ell$. 
By hypothesis, $\alpha$ and $\beta$ are both \old{non-zero} \new{nonzero}, since the block is not aligned. 

By symmetry, assume w.l.o.g. that $\mathcal{G}_k \geq \mathcal{G}_\ell$, meaning that we should try to shift the block \new{to the} left. 
If no other task is scheduled 
during the beginning of $I_{k}$, let $\gamma=0$. Otherwise, let $b_{k}+\gamma$ denote the end time 
of the last task~$v_{q}$ 
that is scheduled before~$v_{r}$. 
Note that $\gamma < \alpha$ because this last task does not belong to the block. We can shift the block left by
$\alpha-\gamma$ time units without exiting $I_{k}$. Similarly, we can shift the block left by
$\beta$ up to reaching the beginning of $I_{\ell}$. 
Hence, we shift the block left by $\delta=\min(\alpha-\gamma,\beta)$ time units. 
After this shift, the block will either be aligned with the beginning of $I_{k}$, or merged with (the block of) 
task~$v_{q}$, or aligned to the beginning of~$I_{\ell}$ (see Figure~\ref{fig.exDP}). 
Also, after this shift, the carbon cost cannot have increased, since we moved some load 
to an interval with a higher green power budget. Since the initial schedule is optimal,
the carbon cost remains the same, the new  
schedule is optimal and it has either one more block aligned, or one less block (or both) 
than the original optimal schedule~$\mathcal{S}$.
Proceeding by induction, this proves that there exists an optimal schedule where all blocks are aligned, 
namely an optimal $\mathcal{E}$-schedule.

\begin{figure}
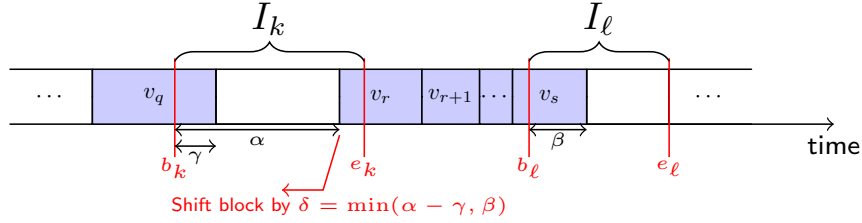

\begin{center}
\resizebox{1.0\textwidth}{!}{ \begin{tikztimingtable}[
    timing/slope=0,         
    timing/rowdist=0.5,     
    timing/coldist=2pt,     
    xscale=3,yscale=2, 	
    thin,              		
]
& 
D{$\cdots$}  {[fill=blue!20] 1.5D{$v_q$}}1.5D{$\;$}{[fill=blue!20] 1D{$v_r$}}{[fill=blue!20] 0.7D{$v_{r+1}$}}{[fill=blue!20] 0.4D{$\cdots$}}{[fill=blue!20] 0.9D{$v_{s}$}} D{$\;$} D{$\cdots$}\\
\extracode
	\makeatletter
 	\begin{pgfonlayer}{foreground}
		\interval{2}{4.3}{0}
		
		\draw [decorate,decoration={brace,amplitude=5pt,raise=1.1ex}]
  (2,.8) -- (4.3,.8) node[midway,yshift=1.6em]{{$I_k$}};

		\interval{6.3}{8}{0}
		\draw [decorate,decoration={brace,amplitude=5pt,raise=1.1ex}]
  (6.3,.8) -- (8,.8) node[midway,yshift=1.6em]{{$I_\ell$}};

			\draw [<->] 
  (2,-.1) -- (4,-.1) node[midway,yshift=-.3em]{\tiny{$\alpha$}};
			\draw [<->] 
  (2,-.4) -- (2.5,-.4) node[midway,yshift=-.3em]{\tiny{$\gamma$}};

			\draw [<->] 
  (6.3,-.1) -- (7,-.1) node[midway,yshift=-.3em]{\tiny{$\beta$}};
  
\draw[<-, color=red!100]  (3.3,-1.2) -- (3.8,-1.2) -- (4,-.2); 
\draw[color=red!100] (4,-1.5) node{\tiny{Shift block by $\delta=\min(\alpha-\gamma,\beta)$}};

%
\draw[color=red!100] (2,-.8) node{\tiny{$b_k$}};
\draw[color=red!100] (4.3,-.8) node{\tiny{$e_k$}};
\draw[color=red!100] (6.3,-.8) node{\tiny{$b_\ell$}};
\draw[color=red!100] (8,-.8) node{\tiny{$e_\ell$}};

 	\end{pgfonlayer}
 	\begin{pgfonlayer}{background}
 		\timeline{1}{10}{0};
 	\end{pgfonlayer}
\end{tikztimingtable} 
}
\end{center}
\caption{Illustration of the block shift to create an  $\mathcal{E}$-schedule. \label{fig.exDP}}
\end{figure}

\paragraph{Restriction to a polynomial number of task end times.}

Given an optimal $\mathcal{E}$-schedule with a single processor, 
each task $v_{u}$  belongs
to a block $[v_{r}, v_{s}]$, where $r \leq u \leq s$. There are $O(n^{2})$ possible such blocks. 
Either $v_{r}$ starts or $v_{s}$ ends (or both) at some time in~$\mathcal{E}$. Hence, the end time of $v_{u}$ can be deduced from the block that it belongs to and the set $\mathcal{E}$ (simply add or subtract the length of previous or following tasks in the block). This leads to $O(n^{2} J)$ possible end times for~$v_{u}$. Summing up over all $n$ tasks, we get $O(n^{3} J)$ possible end times for all the tasks. 
We define $\mathcal{E'}$ as the set of all these possible task end times.

\paragraph{Fully-polynomial dynamic programming algorithm on a single processor.}
Given a schedule on a single processor, we transform it into an $\mathcal{E}$-schedule as explained above,
without increasing its carbon cost. We obtain the set $\mathcal{E'}$ of all possible task end times, which is of cardinality at most $O(n^{3} J)$. 
We can safely restrict the values of $t$ in the pseudo-polynomial dynamic programming algorithm described in Section~\ref{sec.comp.seq} to these  end times, and still derive an optimal schedule. 
This concludes the proof for the design of the fully-polynomial dynamic programming algorithm.

\subsection{Proof of Theorem~\ref{th.severalproc}}
\label{app.comp.para}

In this section, we prove Theorem~\ref{th.severalproc}, namely that  \texttt{UCAS} is strongly NP-complete.
Recall that   \texttt{UCAS} denotes the class of decision problem instances with $P$ processors 
with uniform power consumption, i.e., $\idleP{i} = 0$, $\workP{i} = 1$ for  $1 \leq i \leq P$, and 
independent tasks (no dependence nor communication). Given a bound $C$, we ask whether there exists a valid schedule whose total carbon cost does not exceed~$C$. 

First,  \texttt{UCAS} clearly belongs to NP, a certificate being the schedule itself with the start and end times of the tasks. For the strong completeness, we make a  reduction from the \texttt{3-Partition} problem, which is
strongly NP-complete~\cite{GareyJohnson}.
Consider an instance $\calI_{3P}$ of \texttt{3-Partition} as follows. Let $S=\{x_1,\dots,x_{3n}\}$ be the multiset of $3n$ positive integers, and let $B$ be a bound such that 
    \begin{equation*}
        B = \frac{\sum_{i=1}^{3n}x_i}{n} \quad \text{and} \quad \frac{B}{4} < x_i < \frac{B}{2} \text{~for~} 1 \leq i \leq 3n .
    \end{equation*}
Is there a partition of $S$ into triplets $S=S_1\cup \dots \cup S_{n}$ such that the sum over the 
    elements of each triplet $S_i$ is equal to $B$? We construct the following instance $\calI_{\ucas}$ of
    \ucas:
    \begin{itemize}
    \item We have  $3n$ power-homogeneous processors $p_{1}, \dots, p_{3n}$ with $\idleP{i}=0$ and 
    $\workP{i}=1$ for $1 \leq i \leq 3n$.
    \item We have a workflow of $3n$ independent tasks $v_{1}, \dots, v_{3n}$ with $\omega(v_{i}) = x_{i}$
    for $1 \leq i \leq 3n$. Task $v_{i}$ executes on processor $p_{i}$ for $1 \leq i \leq 3n$.
    \item The horizon is a set of  $J=2n-1$ intervals of total length $T=nB+n-1$. Odd-numbered intervals $I_{1}, I_{3}, \dots, I_{2n-1}$
    each have length $B$ and green power budget $1$, while even-numbered intervals $I_{2}, I_{4}, \dots, I_{2n-2}$ have length $1$ and green budget $0$.
    \item The bound on total carbon cost is $C=0$.
       \end{itemize}
    Clearly, the size of $\calI_{\ucas}$ is linear in the size of $\calI_{3P}$. We show that $\calI_{\ucas}$
    has a solution if and only if $\calI_{3P}$ does.
    
    First, if $\calI_{3P}$ has a solution $S=S_1\cup \dots \cup S_{n}$, we execute each triplet $S_{k}$ in sequence during interval $I_{2k-1}$. The total duration is $B$, and at any time unit, exactly one processor is active. The total cost is indeed $0$, and  $\calI_{\ucas}$ has a solution.
    
    Now, if $\calI_{\ucas}$ has a solution, what is the corresponding schedule? Because the total cost is $0$:
    (i) no task can be executed during even-numbered intervals; and (ii) at most one task can be executed during any time unit of odd-numbered intervals. Because the total length of these intervals is $nB$ (the sum of all task durations), exactly one task is executed at each time unit of odd-numbered intervals. Hence, we have a partition
    of the $3n$ tasks into $n$ subsets of total duration $B$ that are executed sequentially and entirely during these intervals. Because $\frac{B}{4} < x_i < \frac{B}{2}$ for $1 \leq i \leq 3n$, the $n$ subsets are in fact triplets, and we have a solution to $\calI_{3P}$. This concludes the proof.
    
\subsection{Details on the Integer Linear Program}
\label{app:ILP-details}

In this section, we give a detailed formulation for the integer linear program sketched in Section~\ref{sec:ILP}. 
First, we introduce all necessary variables for the model. We need variables for the amount of green (renewable) power
and brown (carbon-emitting) power we use during every time step. Hence, we introduce integer variables 
\begin{equation}
  gu_{t} \geq 0 \quad \text{and} \quad bu_{t} \geq 0,  \quad 0 \leq t < T,
\end{equation} 
for the green power usage and the brown power usage, respectively. Let $G_c = (V_c, E_c,\omega)$ be the communication-enhanced \old{graph} \new{DAG} as defined in Section~\ref{sec.framework}. 
For every (task, \old{mapping} \new{processor}) pair $(v, p_v)$, where  $p_v$ is the processor where 
$v$ is located, we define variables 
\begin{equation}
  s(v,p_v)_{t} \in \{0,1\}, \ e(v,p_v)_{t} \in \{0,1\}, \ r(v,p_v)_{t} \in \{0,1\}, \ 0 \leq t < T, 
\end{equation} 
where $s(v,p_v)_{t} = 1$ if and only if task $v$ is started on processor $p_v$ at time unit $t$, 
$e(v,p_v) = 1$ if and only if task $v$ ends
on processor $p_v$ at time unit $t$, and $r(v,p_v)_t = 1$ if and only if task $v$ is running 
on processor $p_v$ at time unit $t$. For ease of notation, 
we may omit the processor since the mapping is given in advance. Further, we define for 
every time unit $t$ an integer variable $\gamma_{t} \geq 0$, which represents the overall power 
that the cluster uses at time unit~$t$.
In order to compute the brown power usage, we also introduce auxiliary variables $\alpha_t$ for every time unit $t$, where $\alpha_t=1$ if and only if we
need brown power at time unit $t$. With these variables, we can now formulate the integer linear program. 
First, we set the objective function, which writes: 
\begin{equation*}
  \min \sum_{t=0}^{T-1} bu_t .
\end{equation*} 

We now set the constraints for the ILP. First, we need to ensure that every task is executed exactly once, such that it finishes before the 
deadline. This can be modeled by
\begin{align}
  \sum_{t = 0}^{T-\omega(v)} s(v,p_v)_t = 1 \quad \forall v \in V_c, \\
  \sum_{t=T-\omega(v)+1}^{T-1} s(v,p_v)_t = 0 \quad \forall v \in V_c.\label{eq1}
\end{align} Further, we have to ensure the same for the end of the tasks. This is given by \begin{align}
  \sum_{t=0}^{\omega(v)-2} e(v,p_v)_t = 0 \quad \forall v \in V_c, \label{eq2}\\
  \sum_{t=\omega(v)-1}^{T-1} e(v,p_v)_t = 1 \quad \forall v \in V_c.
\end{align} 
Note that if $\omega(v) < 2$,  sums in Eqs.~\eqref{eq1} and \eqref{eq2} are just empty. 

Next, we have to make sure that the beginning of a task and the end of a task are 
aligned according to the execution time of the task. Hence, we get for every task $v \in V_c$ 
the constraints: 
\begin{equation}
  s(v,p_v)_t = e(v,p_v)_{t+\omega(v)-1}, \quad 0 \leq t \leq T-\omega(v) .
\end{equation} 

Lastly, we have to make sure that the variables $r(v,p_v)_t$ are aligned with the beginning and end of the task. Hence, we get for every
task $v$ the constraints: 
\begin{align}
  \sum_{t=0}^{T-1} r(v,p_v)_t = \omega(v), \\
  r(v,p_v)_k \geq s(v, p_v)_t, t \leq k \leq t+\omega(v)-1, \ 0 \leq t < T\new{-\omega(v)}.
\end{align} 

Altogether, we have ensured that every task starts in time. Hence, we now have to make sure that all precedence and communication constraints 
are respected. Due to the structure of the communication-enhanced \old{graph} \new{DAG} $G_c$, 
it is enough to respect every edge as a simple order constraint. Using this graph,
we introduce for every edge $(u, v)\in E_c$ the constraint
\begin{align}
  s(v,p_v)_t \leq \sum_{l=0}^{t-1} e(u, p_u)_l, \quad 0 \leq t < T.
\end{align} 
Note that by construction of $G_c$ these constraints ensure both the internal order of a processor and the inter-processor communication constraints.

Finally, we have to introduce constraints for the power usage. As seen in Section~\ref{sec.framework},
we have to ensure
\begin{align}
  gu_t = \min(\budget_t, \gamma_t), \label{eq:gu} \\
  bu_t = \max(0, \gamma_t - \budget_t),\label{eq:bu}
\end{align} \new{where $\gamma_t$ is the overall power usage at time unit $t$.}
To model Eqs.~(\ref{eq:gu}) and~(\ref{eq:bu}), we use the Big-M method. 
For this, we need auxiliary constants $\epsilon, M \in \mathbb{R}_{>0}$, where
$\epsilon$ is sufficiently small and $M$ is sufficiently large. For $\epsilon$, we can choose any value 
$0 < \epsilon <1$. Further, it suffices to estimate 
$M$ by 
\begin{equation*}
  M \geq \sum_{i=1}^{P^2} \left( \idleP{i} + \workP{i}\right), 
\end{equation*} 
since we cannot use more brown power than that.

First, we ensure Eq.~(\ref{eq:bu}) by the following constraints:
\begin{align}
  bu_t \geq 0, \quad 0 \leq t < T\\
  bu_t \geq \gamma_t - \budget_t, \quad 0 \leq t < T\\
  bu_t \leq \gamma_t - \budget_t + M(1-\alpha_t)\quad 0 \leq t < T\\
  bu_t \leq M\cdot \alpha_t , \quad 0 \leq t < T\\
  \gamma_t - \budget_t \leq M\cdot \alpha_t, \quad 0 \leq t < T\\
  \gamma_t - \budget_t \geq \epsilon - M(1-\alpha_t)\quad 0 \leq t < T.
\end{align} 
\new{Here, $\alpha_t$ is a binary variable indicating whether we need brown power or not.} To determine the green power usage, we \new{then} only need the constraints 
\begin{align}
  gu_t \geq 0, \quad 0 \leq t < T\\
  gu_t + bu_t = \gamma_t ,\quad 0 \leq t < T.
\end{align}
Finally, we must determine the overall power usage by the constraint:
 \begin{equation}
  \gamma_t = \sum_{i=1}^{P^2} \idleP{i} + \sum_{v\in V_c} r(v,p_v)_t \new{\cdot \workP{p_v}} , \quad 0 \leq t < T.
\end{equation}

\subsection{Further Experimental Results}
\label{sec:moreResults}
\begin{figure}[H]
  \includegraphics[width=\linewidth]{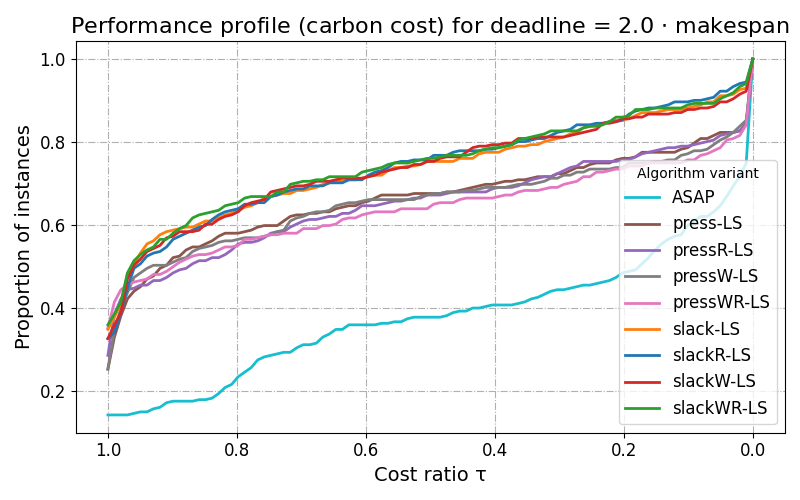}    
  \caption{Performance profile for tolerance factor $2$ for the deadline (extends Figure~\ref{fig.pf_offset}).}
  \label{fig:pf_offset_2}
\end{figure}

\begin{figure}[H]
  \includegraphics[width=\linewidth]{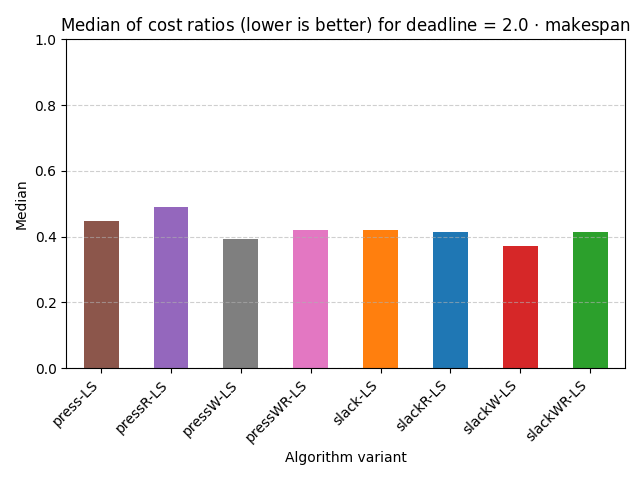}    
  \caption{Cost ratios for tolerance factor $2$ for the deadline (extends Figure~\ref{fig.means_offset}).}
  \label{fig:pf_means_2}
\end{figure}

\begin{figure}[H]
  \centering
  \includegraphics[width=\linewidth]{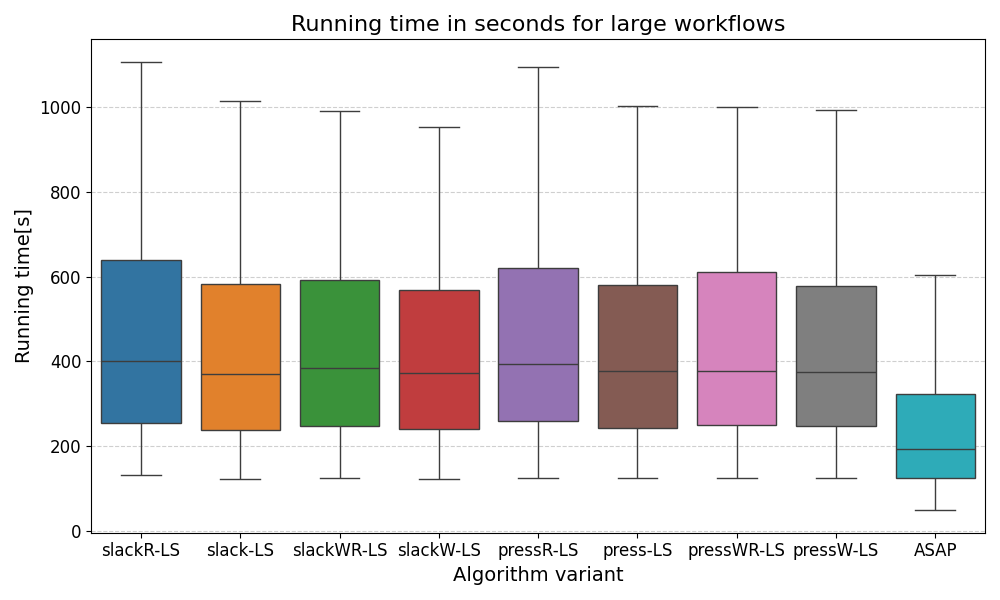}
  \caption{Time (in seconds) for each algorithm variant only for large workflows. Large workflows have between \numprint{20000} and \numprint{30000} tasks.}
  \label{fig.time_comparison_largeWF}
\end{figure}

\begin{figure*}
  \centering
  \begin{subfigure}[b]{0.6\textwidth}
    \includegraphics[width=\linewidth]{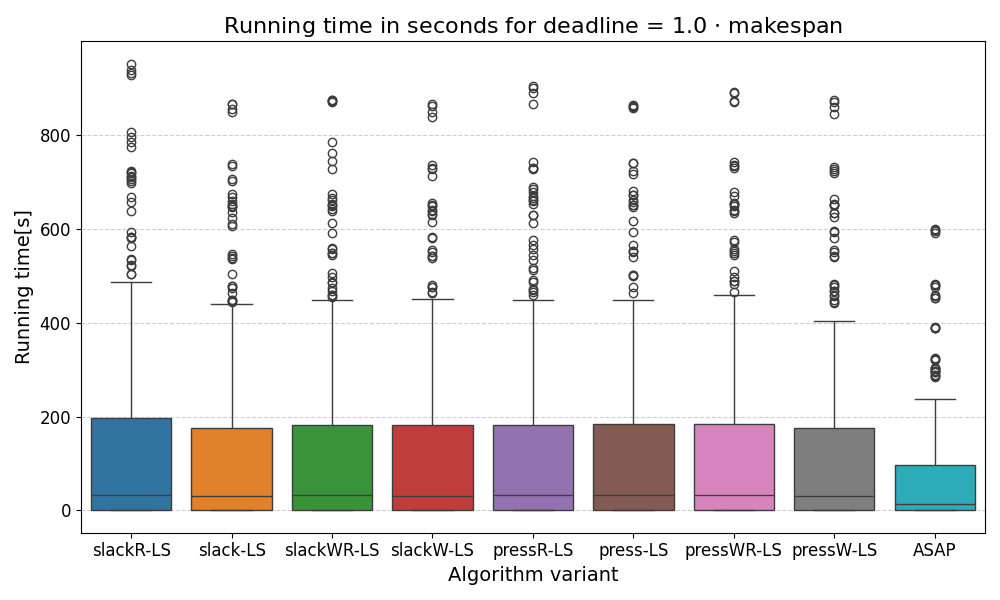}    
  \end{subfigure}
  \hfill 
  \begin{subfigure}[b]{0.6\textwidth}
    \includegraphics[width=\linewidth]{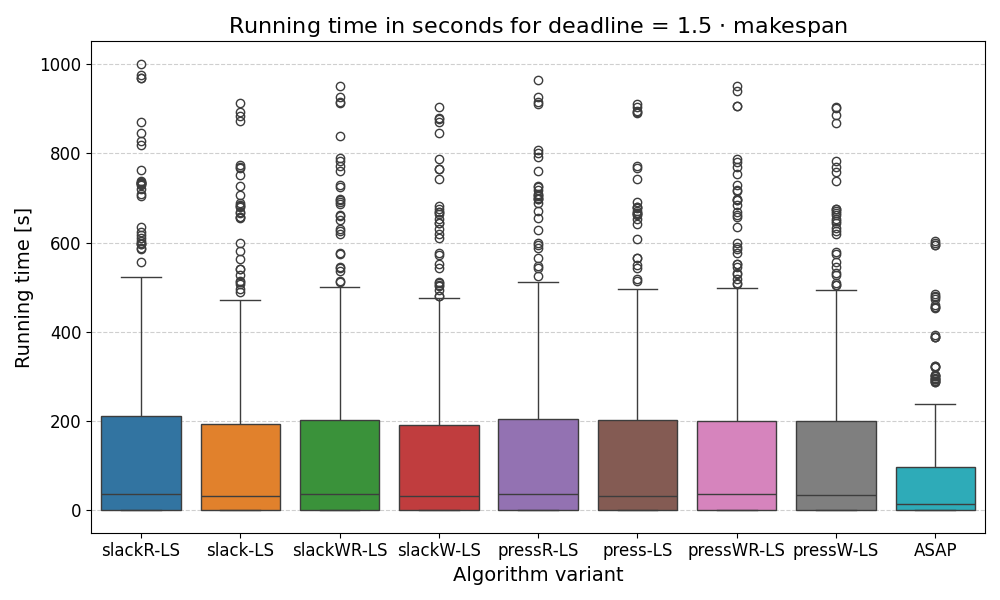}    
  \end{subfigure}
  \hfill 
   \begin{subfigure}[b]{0.6\textwidth}
    \includegraphics[width=\linewidth]{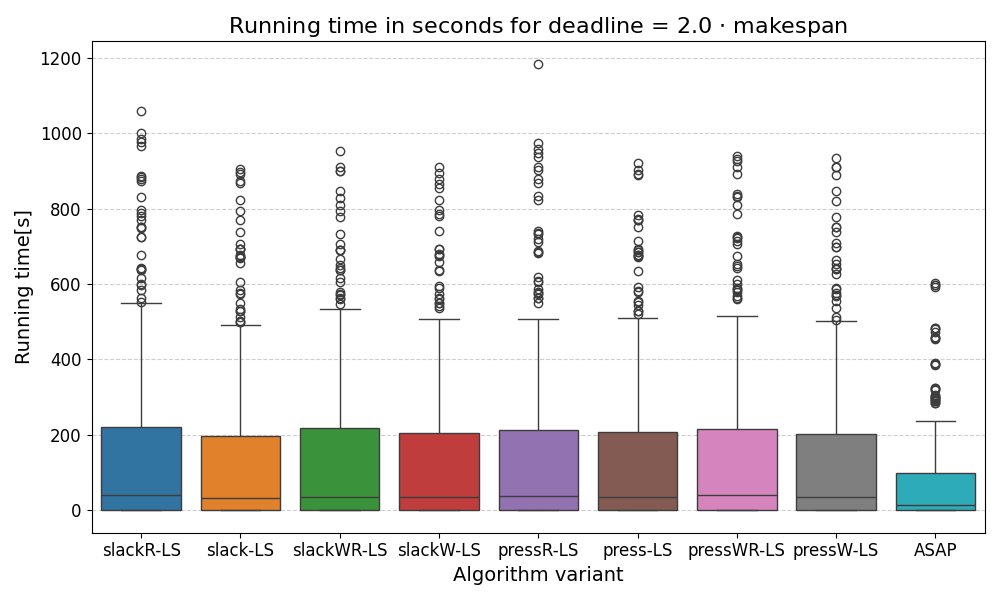}    
  \end{subfigure}
  \hfill 
  \begin{subfigure}[b]{0.6\textwidth}
    \includegraphics[width=\linewidth]{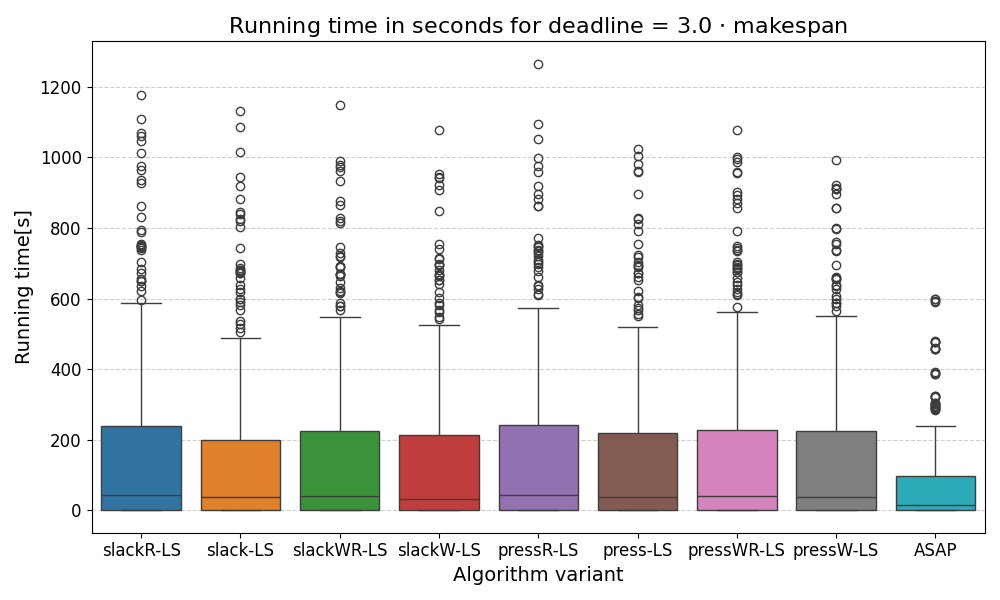}    
  \end{subfigure}
  \caption{The evolution of the running time when adding more tolerance to the deadline from top to bottom.}    
  \label{fig.rt_offset}
\end{figure*}

\begin{figure*}
  \centering
  \begin{subfigure}[b]{1.0\textwidth}
    \includegraphics[width=\linewidth]{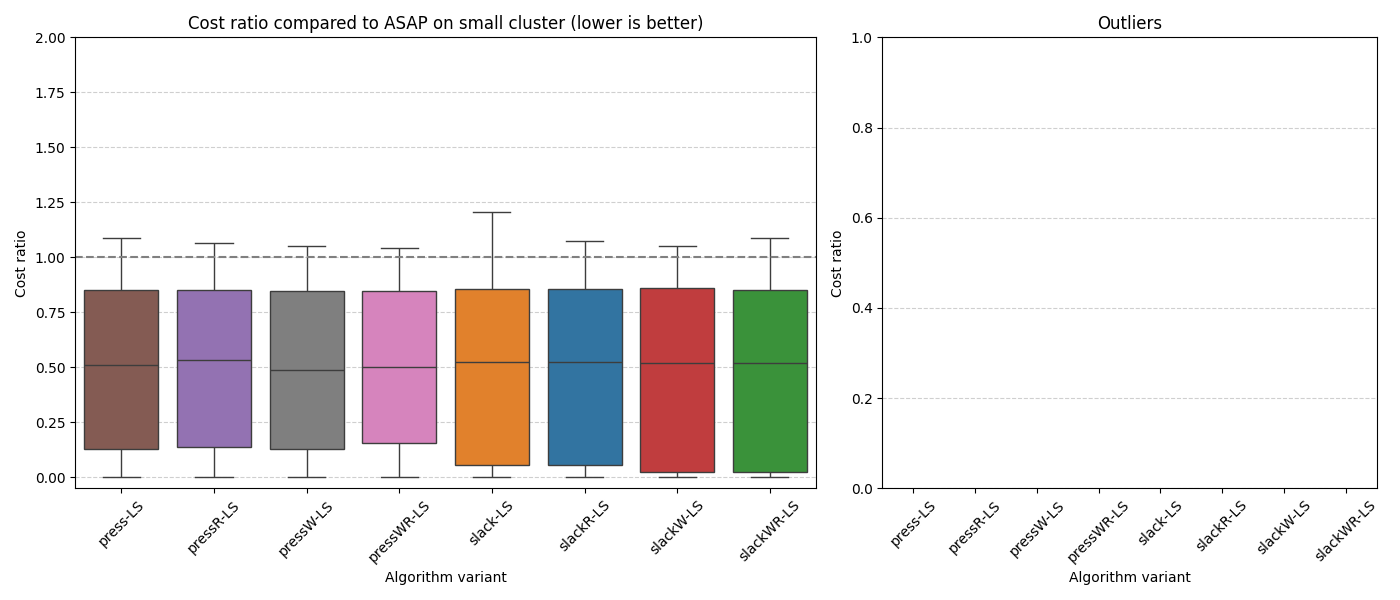}    
  \end{subfigure}
  \hfill 
  \begin{subfigure}[b]{1.0\textwidth}
    \includegraphics[width=\linewidth]{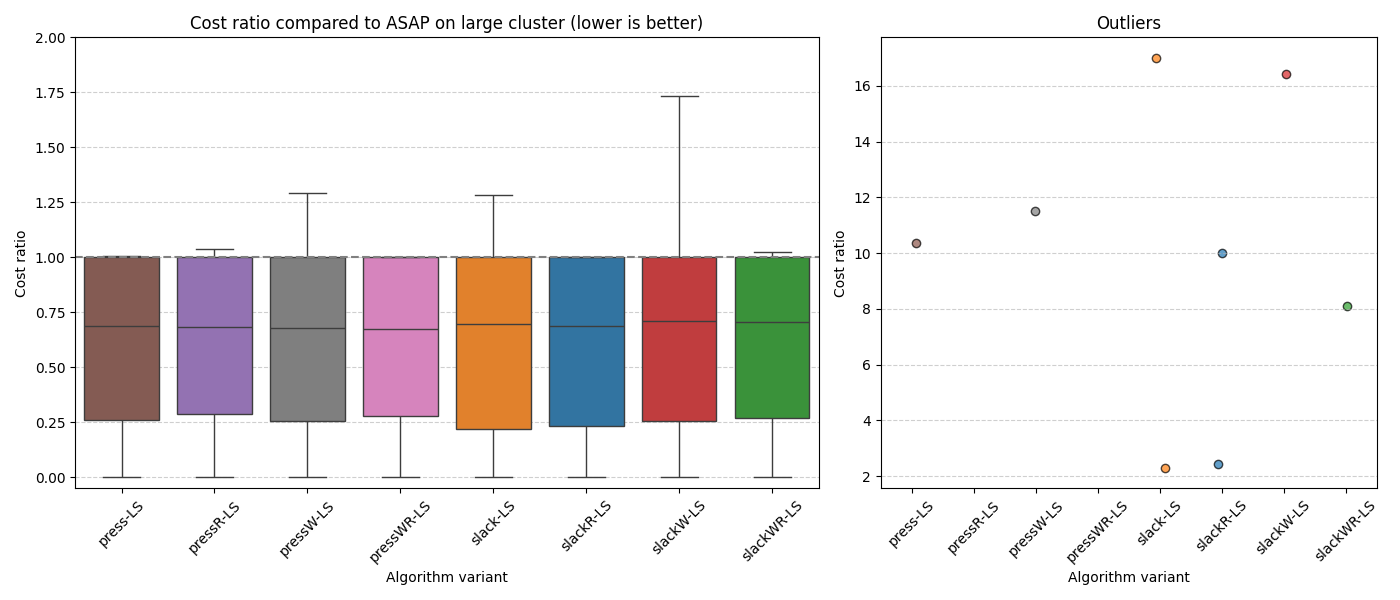}    
  \end{subfigure}
  \hfill 
  \caption{Cost ratio for different cluster sizes.}    
  \label{fig.costRatio_cluster}
\end{figure*}

\begin{figure*}
  \centering
  \begin{subfigure}[b]{0.9\textwidth}
    \includegraphics[width=\linewidth]{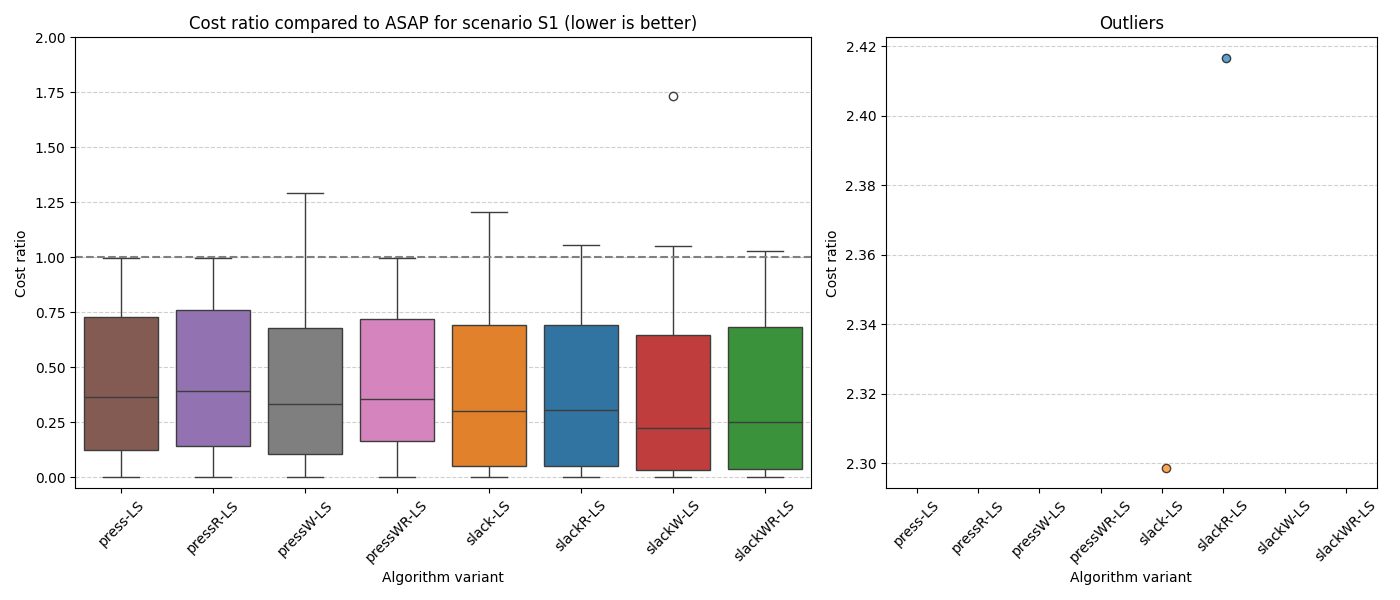}    
  \end{subfigure}
  \hfill 
  \begin{subfigure}[b]{0.9\textwidth}
    \includegraphics[width=\linewidth]{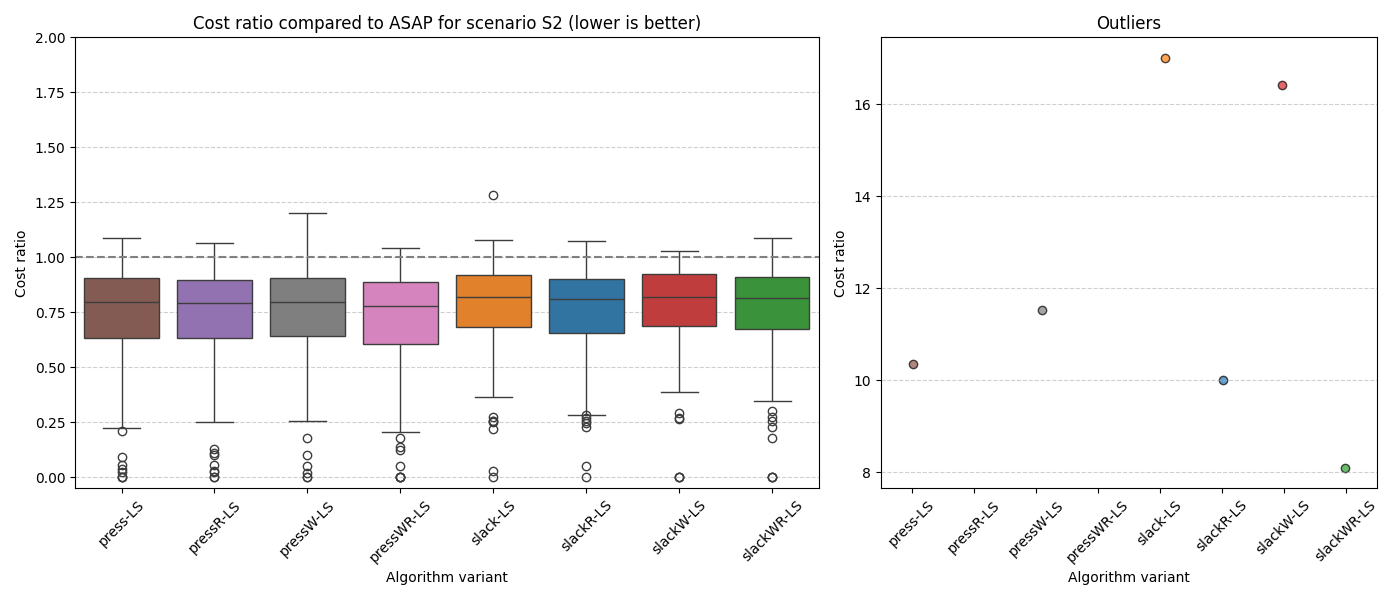}    
  \end{subfigure}
  \hfill 
   \begin{subfigure}[b]{0.9\textwidth}
    \includegraphics[width=\linewidth]{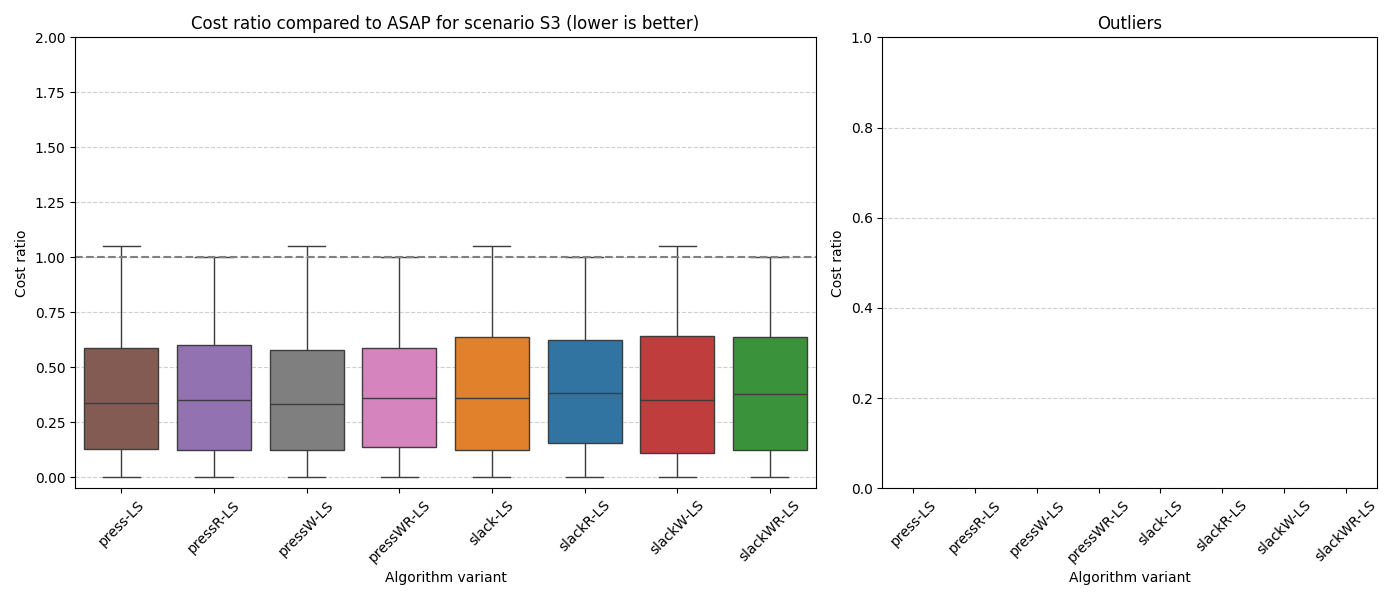}    
  \end{subfigure}
  \hfill 
  \begin{subfigure}[b]{0.9\textwidth}
    \includegraphics[width=\linewidth]{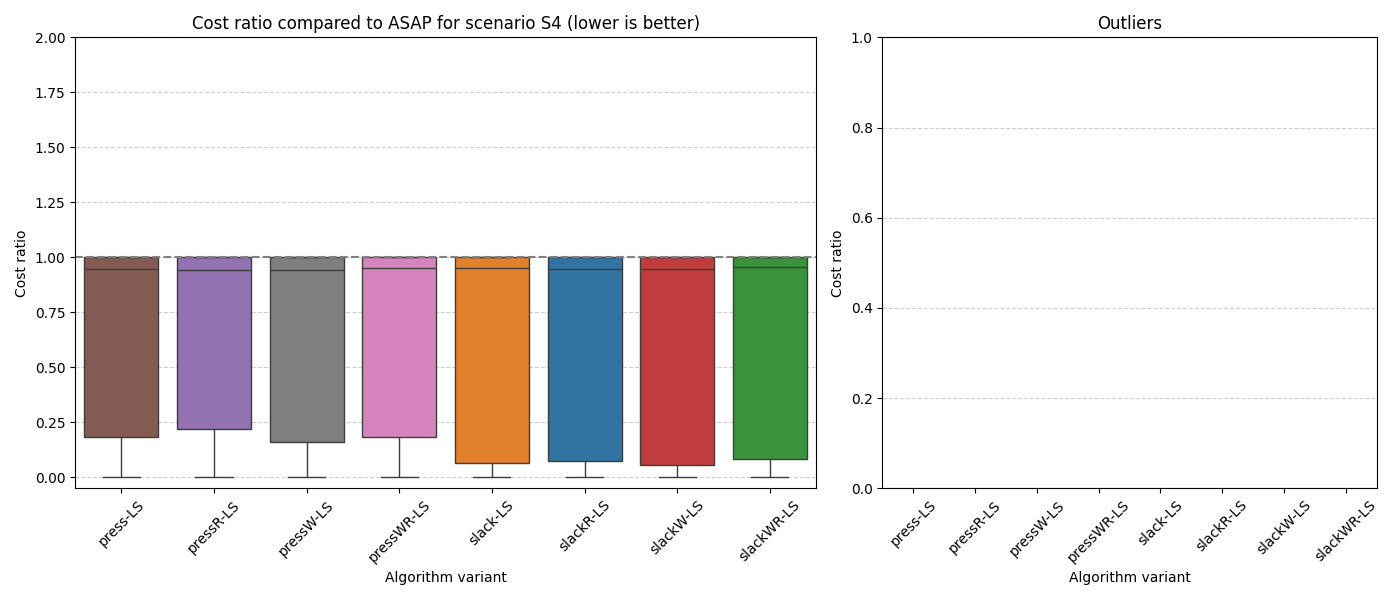}    
  \end{subfigure}
  \caption{Cost ratio for different power profiles.}    
  \label{fig.costRatio_profile}
\end{figure*}

\begin{figure*}
  \centering
  \begin{subfigure}[b]{0.85\textwidth}
    \includegraphics[width=\linewidth]{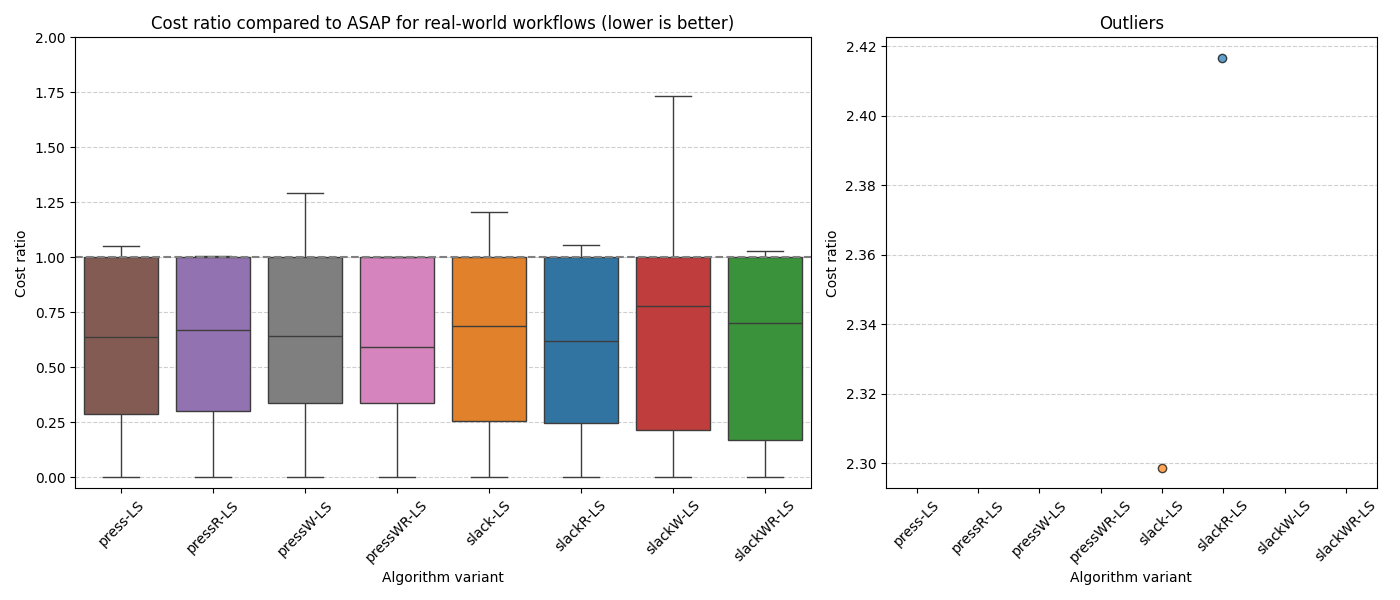}    
  \end{subfigure}
  \hfill 
  \begin{subfigure}[b]{0.85\textwidth}
    \includegraphics[width=\linewidth]{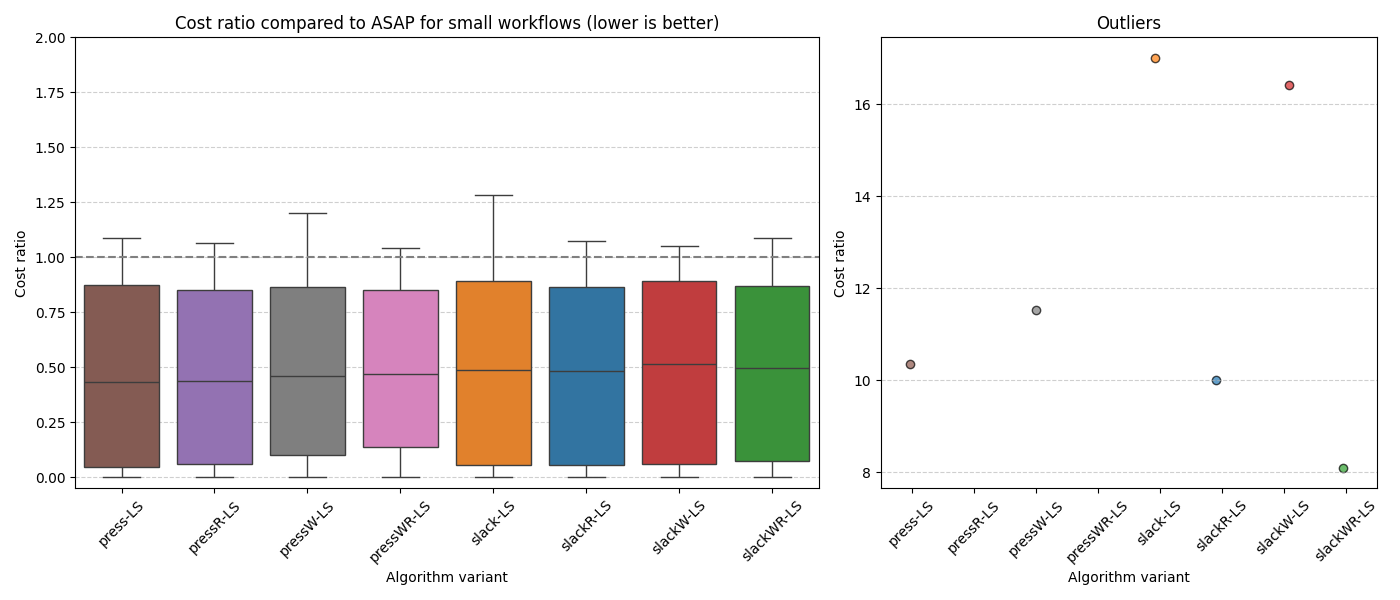}    
  \end{subfigure}
  \hfill 
   \begin{subfigure}[b]{0.85\textwidth}
    \includegraphics[width=\linewidth]{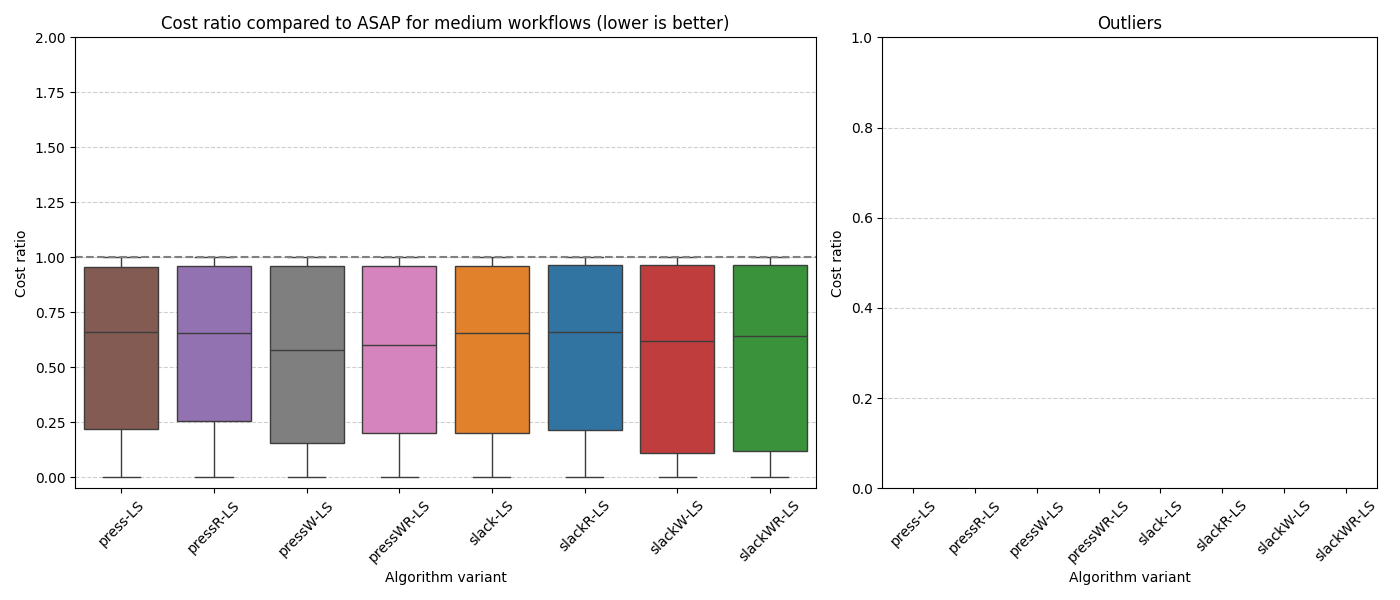}    
  \end{subfigure}
  \hfill 
  \begin{subfigure}[b]{0.85\textwidth}
    \includegraphics[width=\linewidth]{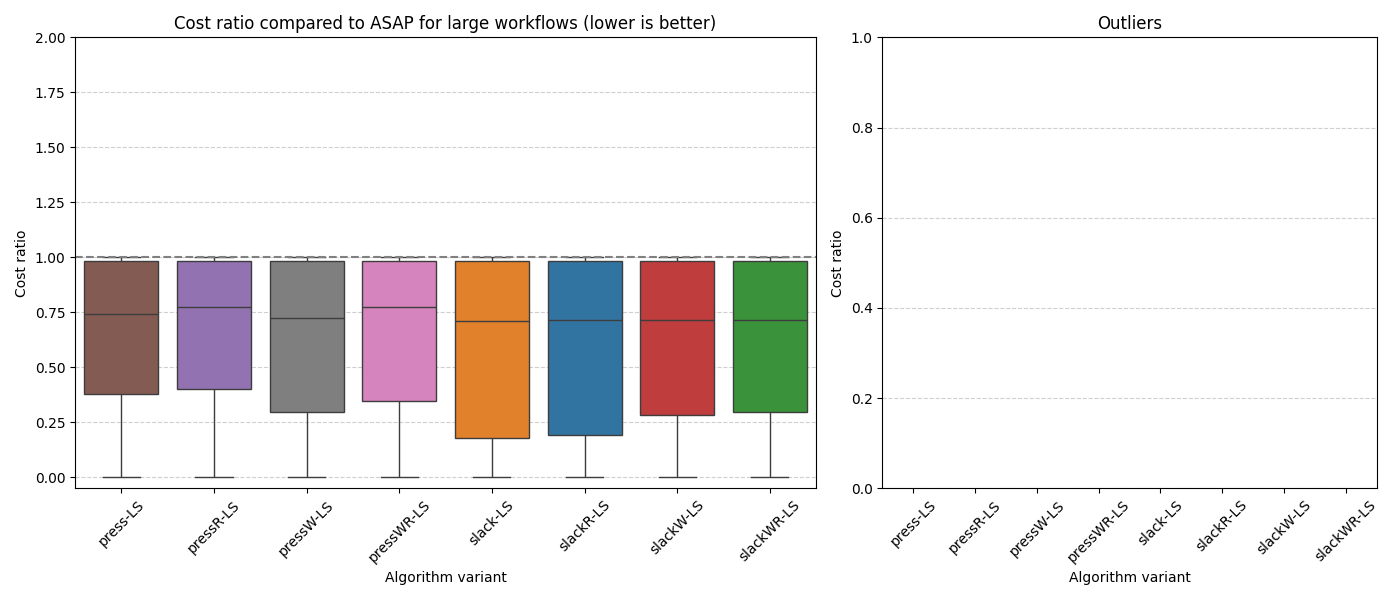}    
  \end{subfigure}
  \caption{Cost ratio for different sized workflows. Small workflows have between \numprint{200} and \numprint{4000} tasks, 
  medium workflows have between \numprint{8000} and \numprint{18000} tasks, and large workflows have between \numprint{20000} and \numprint{30000} tasks.}    
  \label{fig.costRatio_workflows}
\end{figure*}

\begin{figure*}
  \centering
  \begin{subfigure}[b]{1.0\textwidth}
    \includegraphics[width=\linewidth]{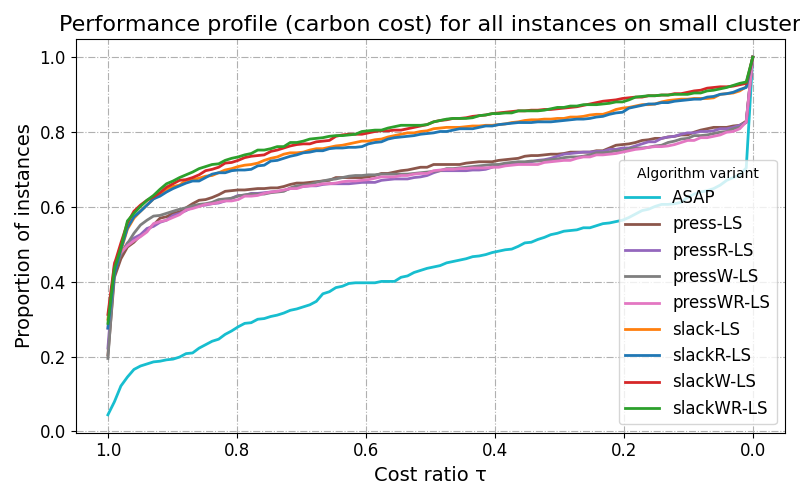}    
  \end{subfigure}
  \hfill 
  \begin{subfigure}[b]{1.0\textwidth}
    \includegraphics[width=\linewidth]{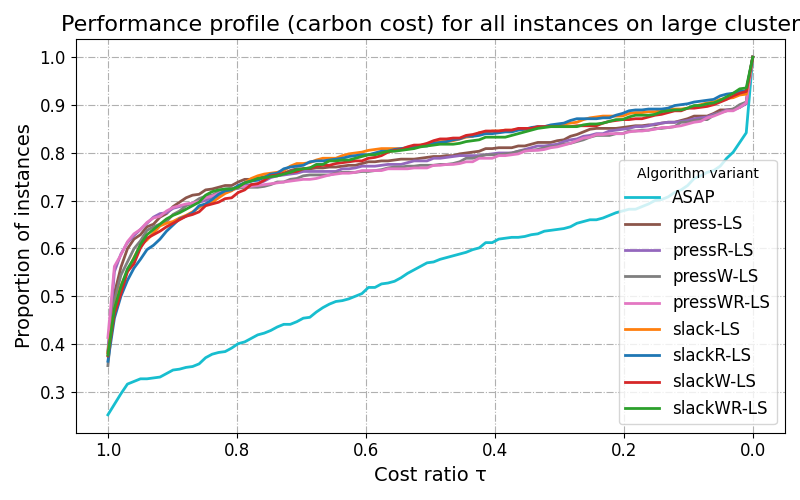}    
  \end{subfigure}
  \hfill 
  \caption{Performance profile for different cluster sizes.}    
  \label{fig.pf_cluster}
\end{figure*}

}{}

\end{document}